\documentclass[10pt,journal]{IEEEtran}
\usepackage[unicode=true,pdfusetitle, bookmarks=true,bookmarksnumbered=false,bookmarksopen=false, breaklinks=false,backref=false,colorlinks=false]{hyperref}
\hypersetup{colorlinks,linkcolor=myurlcolor,citecolor=myurlcolor,urlcolor=myurlcolor}
\usepackage{graphics,epstopdf,graphicx,amsthm,amsmath,amssymb,mathptmx,braket,colortbl,color,bm,framed,mathrsfs}
\usepackage[up]{subfigure}
\usepackage{tikz}
\usepackage{float}
\usepackage{cite}

\usepackage{hyperref}
\usepackage[T1]{fontenc} 
\usepackage{amsmath}
\usepackage{mathtools}
\usepackage[cmintegrals]{newtxmath}

\usepackage{bm} 

\usepackage[normalem]{ulem}

\usepackage{xcolor,cancel}

\usepackage{multirow,tabu}
\usepackage{array}
\usepackage{tabularx}
\newcolumntype{L}[1]{>{\raggedright\let\newline\\\arraybackslash\hspace{0pt}}m{#1}}
\newcolumntype{C}[1]{>{\centering\let\newline\\\arraybackslash\hspace{0pt}}m{#1}}
\newcolumntype{R}[1]{>{\raggedleft\let\newline\\\arraybackslash\hspace{0pt}}m{#1}}
\usepackage{tabu}

\usepackage{qcircuit}

\definecolor{myurlcolor}{rgb}{0,0,0.9}
\newcommand{\be}{\begin{equation}}
\newcommand{\ee}{\end{equation}}
\newcommand{\beq}{\begin{eqnarray}}
\newcommand{\eeq}{\end{eqnarray}}
\newcommand{\beqs}{\begin{eqnarray*}}
\newcommand{\eeqs}{\end{eqnarray*}}

\theoremstyle{plain}
\newtheorem{thm}{Theorem}
\newtheorem{lem}[thm]{Lemma}
\newtheorem{prop}[thm]{Proposition}
\newtheorem{cor}[thm]{Corollary}
\theoremstyle{definition}
\newtheorem{Def}[thm]{Definition}

\tikzstyle WL=[line width=10pt,opacity=1.0]
\newcommand*{\myproofname}{Proof}

\allowdisplaybreaks
\makeatother

\begin{document}

\title{Privacy-Utility Trade-Off}

\author{Hao Zhong and Kaifeng Bu
\thanks{Hao Zhong is with the College of Information Science and Technology, Chengdu University of Technology, Chengdu, 610059, China,(E-mail: zhonghao@cdut.edu.cn)}
\thanks{Kaifeng Bu is with Department of Physics, Harvard University, Cambridge, Massachusetts 02138, USA,(E-mail: kfbu@fas.harvard.edu.)}}

\maketitle

\begin{abstract}
In this paper, we investigate the privacy-utility trade-off (PUT) problem, which considers the minimal privacy loss at a fixed expense of utility. Several different  kinds of privacy in the PUT problem are studied, including differential privacy, approximate differential privacy, maximal information, maximal leakage, R\'{e}nyi differential privacy, Sibson mutual information and mutual information. The average Hamming distance is used to measure the distortion caused by the privacy mechanism. 
We consider two scenarios: global privacy and local privacy.
In the framework of  global privacy framework, the privacy-distortion function is upper-bounded by the privacy loss of a special mechanism, and lower-bounded by the optimal privacy loss with any possible prior input distribution. In  the framework of local privacy, we generalize a coloring method for the PUT problem. 

\end{abstract}

\begin{IEEEkeywords}
Privacy-utility trade-off, differential privacy, maximal leakage, R\'{e}nyi differential privacy, Sibson mutual information
\end{IEEEkeywords}

\section{Introduction}

Advance in digital technologies of data-producing and data-using provides enormous reliable and convenient services in the era of Big Data. Along with the prosperity, security issues and challenges result in increasing privacy concerns. Nowadays, examples include patient data breach, facial recognition data abuse, and surveillance measures without consent during the COVID-19 pandemic. 

Traditional semantic security for cryptosystems is not always achievable due to the auxiliary information available to adversary. Hence, differential privacy has been introduced by Dwork et al. \cite{Dwork2006a, Dwork2006b} to
deal with this problem. The original differential privacy is a typical divergence-based privacy since it is exactly the max-divergence between the probability distributions on neighboring databases. 

 This new measure requires an upper bound on the worst-case change, which may take high expense of utility. As a result, several relaxations of differential privacy have been studied recently. We categorize these existing measures into two main classes, divergence-based privacy and mutual information-based privacy. Besides, $\delta$-approximation can be used for particular privacy notion to relax the strict relative shift.

The most generalized divergence-based privacy for now was given by \cite{Mironov2017}, R\'{e}nyi differential privacy. This $\alpha$-R\'{e}nyi divergence-based measure keeps track of the privacy cost and allows tighter analysis of mechanism composition \cite{Papernot2016, Papernot2018, Feldman2021}. All the divergence-based measures consider the worst-case guarantee of privacy, and thus they protect the information of every single individual.

The mutual information-based privacy concerns the expected information leakage from the input (plain) database to the output (synthetic or sanitized) database. A straightforward measure is based on Shannon’s
mutual information which makes use of related results in information theory \cite{Cuff2016, Wang2016}. Another measure called max-information was also given by Dwork et al. \cite{Dwork2015}, which aims at bounding the change in the conditional probability of events relative to their a priori probability. Two most common mutual information-based privacy measures are based on Sibson's \cite{Sibson1969} and Arimoto's $\alpha$-mutual information \cite{Arimoto1975}, respectively. Plenty of works \cite{Verdu2015, Ho2015, Liao2019, Wu2020, Esposito2020a, Esposito2021} based on the $\alpha$-mutual information were published, concerning data processing and mechanism composition. Notably, the case $\alpha=\infty$ leads to a newly defined privacy definition, maximal leakage, which was first proposed by Issa et al. \cite{Issa2016}. As compared to most of mutual-information metrics, it is easier to compute and analyze.

In Ref.~\cite{Dwork2006c}, Dwork et al. defined approximate $(\epsilon,\delta)$-diferential privacy. It is a differential privacy notion with a $\delta$ shift tolerance, which provides more flexibility while designing mechanisms. In another paper by Dwork \cite{Dwork2015}, she unified this idea with max information. Let $x$ and $x^{\prime}$ be two adjacent database, i.e., differing in only one individual. It is worth noting that for $(\epsilon,\delta)$-differential privacy, the mechanism $Q$ is approximate private if and only if there exists a distribution which is statistically $\delta$-close to $Q(x)$ and max-divergently $\epsilon-$close to $Q(x^{\prime})$ \cite{Dwork2014}. 

 The main goal is to find a privacy-preserving mechanism not destroying statistical utility, or to solve the privacy-utility trade-off (PUT) problem. A rate-distortion theory flavored PUT problem is to determine the minimal compression during a valid communication \cite{Kalantari2018}. As an important topic in information theory, rate distortion has been studied and applied in many fields including complexity and learning theory. For example, the information bottleneck (IB) problem which was initiated by Tishby \cite{Tishby1999}, has attracted lots of attention as it provides an explanation of the success of machine learning \cite{Tishby2015, Shwartz-Ziv2017, Saxe2018, Goldfeld2020}. As for PUT problem, global differential privacy-distortion with average Hamming distance was studied in \cite{Wang2016}; local differential privacy-distortion with average Hamming distance was studied in \cite{Kalantari2018}; local maximal leakage-distortion was studied in \cite{Saeidian2021}; Arimoto mutual information-distortion with hard Hamming distance was studied in \cite{Liao2019}.

In this paper, we study the the minimal privacy loss with given Hamming distortion for seven
different privacy notions
including differential privacy, approximate differential privacy, maximal information, maximal leakage, R\'{e}nyi differential privacy, Sibson mutual information and mutual information.  Our main contributions in this work are listed as follows:
\begin{enumerate}
    \item We obtain the lower- and upper-bounds of privacy-distortions for the privacy notions mentioned above. Moreover, for the prior input distribution is unknown, we obtain the closed-form expression for differential privacy- and rate-distortion.
    \item We generalize the coloring method to solve the PUT problems for local approximate differential privacy- and maximal leakage-distortion. The term "local" means the data server is untrusted and the individual information is calibrated before uploading. 
    \item If the input space consisting of samples independently and identically distributed chosen, we show the analytical closed form results for differential privacy-, maximal leakage- and rate-distortion, together with the lower- and upper-bounds for the rest of privacy-distortions. 
\end{enumerate}

The remainder of this paper is organized as follows. In Section \ref{sec:2},
we introduce the basics notions and concepts. In Section \ref{sec:3}, we investigate the global privacy distortion, i.e., the privacy-utility trade-off when there is a curator with access to the plain data. In Section \ref{sec:4}, we study the local privacy distortion and demonstrate our results by two images. Moreover,  we show the privacy-distortion trade-off for parallel composition. Finally, we conclude in Section \ref{sec:6}.
 

\section{Preliminaries}\label{sec:2}
Let  $\mathcal{X}$ be the input data space, which is a finite set. A randomized mechanism $Q$ is a conditional distribution $Q_{Y|X}$ mapping the discrete random variable $X$ to an output data $Y\in\mathcal{Y}.$ We assume that the data $x$ is drawn from some distribution $P_X$ with full support on $\mathcal{X}$. 
For simplicity, we denote $P_X(i)$ and $Q_{Y|X}(j|i)$ as $P_i$ and $Q(j|i)$ for $i\in\mathcal{X}$ and $j\in\mathcal{Y}$, respectively.

In general, the prior distribution $P_X$ is unknown. Hence, we assume it belongs to some set of probability distribution on $\mathcal{X}$, which is called source set of $P_X$ and discloses the prior knowledge about the input distribution. 
For example, Kalantari et al. \cite{Kalantari2018} divided the set of source sets into the following three classes.

\begin{Def}[\cite{Kalantari2018}]
A source set $\mathcal{P}$ is classified into three classes as follows:

\textbf{Class \uppercase\expandafter{\romannumeral1}.} Source set $\mathcal{P}$ whose convex hull Conv$(\mathcal{P})$ includes the uniform distribution $P_U$. 

\textbf{Class \uppercase\expandafter{\romannumeral2}.} Source set $\mathcal{P}$ not belonging to Class \uppercase\expandafter{\romannumeral1}, with ordered probability values. Without loss of generality, $P_1\geq P_2\geq \cdots \geq P_{|\mathcal{X}|}$.

\textbf{Class \uppercase\expandafter{\romannumeral3}.} Source set $\mathcal{P}$ that cannot be classified as Class I or Class II.
\end{Def}

Note that Class I contains two special cases: (1) it only contains the  uniform distribution and (2) it contains 
all possible probability distributions, which means there is no auxiliary information about input distribution. And Class II contains a special case that it only contains one non-uniform probability distribution. 

\subsection{Distortion Measure}
The input space $\mathcal{X}$ can represent attributes for a single individual or several individuals. A straight way to estimate the utility of mechanism is to calculate the distortion between $X$ and $Y$, where smaller distortion corresponds to greater usefulness. This is because smaller distortion implies one can gain more information of the original data from the released one. Let $\mathcal{Y}$ be a synthetic database, i.e., it is in the same universe with $\mathcal{X}$. Then we define the distortion as follows.
\begin{eqnarray}
\mathbb{E}_{P,Q}[d(X,Y)]=\sum_{x\in\mathcal{X}}\sum_{y\in\mathcal{Y}}P_X(x)Q_{Y|X}(y|x)d(x,y),
\end{eqnarray}
where $d(x,y)$ is the Hamming distance between $x$ and $y$, i.e., the number of attributes they differ in.  Two elements $x$ and $x^{\prime}$ are called neighbors if $d(x,x^{\prime})=1$. In this paper, we focus on the accuracy-persevering mechanisms that do not distort the data above some threshold $D$. 
\begin{Def}
A mechanism $Q$ is $(P,D)$-valid if the expected distortion $\mathbb{E}_{P,Q}[d(x,y)]\leq D$, where $\mathbb{E}_{P,Q}$ denotes the average over the randomness of input and mechanism. Moreover, $Q$ is $(\mathcal{P},D)$-valid if $Q$ is $(P,D)$-valid for any $P\in \mathcal{P}$. The set of all $(P,D)$-valid ($(\mathcal{P},D)$-valid, respectively) mechanisms is denoted by $\mathcal{Q}(P,D)$ ($\mathcal{Q}(\mathcal{P},D)$, respectively).
\end{Def}

Due to the convexity, it is directly to see that $\mathcal{Q}(\mathcal{P},D)=\mathcal{Q}(Conv(\mathcal{P}),D)$. Hence, without loss of generality, we always assume that $\mathcal{P}$ is convex.

\subsection{Different Notions of Privacy}
In this paper, we investigate different notions of privacy, which we list as follows.
\begin{Def}\label{def:prav}
Let $P_X$ be a distribution with full support on $\mathcal{X}$ and let $\epsilon$ represent a non-negative real number.

(1) [$\epsilon$-differential privacy, \cite{Dwork2006a, Dwork2006b}] A mechanism $Q$ is called $\epsilon$-differential private if for any $y\in\mathcal{Y}$, neighboring elements $x$ and $x^{\prime}\in\mathcal{X}$,
$$Q_{Y|X}(y|x)\leq e^{\epsilon} Q_{Y|X}(y|x^{\prime}).$$
In this case, we denote the minimal $\epsilon$ as $\epsilon_{DP}(Q)$ for the mechanism $Q$. 

(2) [$(\epsilon,\delta)$-differential privacy, \cite{Dwork2014}] Given a positive  number $\delta<1$, a mechanism $Q$ is called $(\epsilon,\delta)$-differential private if for any $y\in\mathcal{Y}$, neighboring elements $x$ and $x^{\prime}\in\mathcal{X}$,
$$Q_{Y|X}(y|x)\leq e^{\epsilon} Q_{Y|X}(y|x^{\prime})+\delta.$$
In this case, we denote the minimal $\epsilon$ as $\epsilon_{\delta}(Q)$ for the mechanism $Q$. 

(3) [Maximal information, \cite{Dwork2015}] A mechanism $Q$ has maximal information $\epsilon$ if for any $y\in\mathcal{Y}$, $x\in\mathcal{X}$,
$$Q_{Y|X}(y|x)\leq e^{\epsilon} \mathop{\mathbb{E}}_{x\sim P_X}[Q_{Y|X}(y|x)].$$
In this case, we denote the minimal $\epsilon$ as $\epsilon_{MI}(Q)$ for the mechanism $Q$. 

(4) [Maximal leakage, \cite{Issa2016}] A mechanism $Q$ has maximal leakage $\epsilon$ if
$$\sum_{y\in\mathcal{Y}}\max_{x\in\mathcal{X}}Q_{Y|X}(y|x)\leq e^{\epsilon}.$$
In this case, we denote the minimal $\epsilon$ as $\epsilon_{ML}(Q)$ for the mechanism $Q$. 

(5) [R\'{e}nyi differential privacy, \cite{Mironov2017}] Given a positive number $\alpha>1$, a mechanism $Q$ is $\epsilon$-R\'{e}nyi differential private if for any neighboring elements $x$ and $x^{\prime}\in\mathcal{X}$, the R\'{e}nyi divergence
$$D_{\alpha}\left({Q_{Y\vert X}(\cdot\vert x)\Vert Q_{Y\vert X}(\cdot\vert x^{\prime})}\right)\leq \epsilon.$$
In this case, we denote the minimal $\epsilon$ as $\epsilon_{\alpha,DP}(Q)$ for the mechanism $Q$. 

(6) [$\alpha$-mutual information privacy, \cite{Esposito2020a, Sibson1969, Verdu2015}] Given a positive number $\alpha>1$, a mechanism $Q$ has maximal $\alpha$-mutual information $\epsilon$ if the the Sibson mutual information 
$$I_{\alpha}\left({X;Q(X)}\right)\leq \epsilon.$$
In this case, we denote the minimal $\epsilon$ as $\epsilon_{\alpha,M}(Q)$ for the mechanism $Q$. 

(7) [Mutual information privacy, \cite{Cuff2016, Wang2016}] A mechanism $Q$ is called $\epsilon$-mutual information private if the mutual information 
$$I\left({X;Q(X)}\right)\leq\epsilon.$$
In this case, we denote the minimal $\epsilon$ as $\epsilon_{M}(Q)$ for the mechanism $Q$. 

\end{Def}

The minimum $\epsilon_{\star}(Q)$ above is called the privacy loss of mechanism $Q$. We now introduce the privacy-distortion function, which describes the minimal achievable privacy loss under restricted distortion. It is defined as follows. 
\begin{Def}\label{def:pri-dis}
Let $\mathcal{P}$ be a set with full-support distributions and $D$ a positive number.
\begin{equation}
\epsilon^{*}_{\star}(\mathcal{P},D):=\min_{Q\in\mathcal{Q}(\mathcal{P},D)}\epsilon_{\star}(Q)   
\end{equation}
is called the corresponding privacy-distortion function for $\star$ representing the privacy notions aforementioned. The set of mechanisms achieving the minimum is denoted by $\mathcal{Q}^{*}_{\star}(\mathcal{P},D)$. In particular, if $\mathcal{P}$ consists of only a single element $P$, then we write them as $\epsilon^{*}_{\star}(P,D)$ and $\mathcal{Q}_{\star}^{*}(P,D)$. Conversely, given restricted privacy cost $\epsilon\geq0$, we denote the smallest expected distortion achievable by 
\begin{equation}
D_{\star}^{*}(\mathcal{P},\epsilon):=\min_{\epsilon_{\star}(Q)\leq\epsilon}\max_{P\in\mathcal{P}}\mathbb{E}_{P,Q}[d(x,y)].
\end{equation}
\end{Def}

\section{Global Privacy Framework}\label{sec:3}

A global privacy framework relies on a data server which collects all the data directly from individuals and publish them in a privacy-protected synthetic database (see Fig. \ref{fig:global}). The $n$-individual database $x$ is drawn from $\mathcal{X}=\{1,2,\cdots,m\}^{n}$ with respect to the probability distribution $P$ in $\mathcal{P}$ where $m$ is the total number of types the individuals belong to. 
For fixed $x$, all database in $\mathcal{X}$ can be divided into $n+1$ categories based on their Hamming distance to $x$. Let $\mathcal{N}_{l}(x)$ be the set consisting of all elements with $l$ attributes different from $x$ and $0\leq l\leq n$. Then the size of the set $\mathcal{N}_{l}(x)$ is $N_{l}:=\binom{n}{l}(m-1)^l$. 
To avoid notation representation confusion, 
$\tilde{\epsilon}^{*}$ is used to denote the privacy-distortion function of global case.
Here, we have the following result which holds for every privacy notions in Definition \ref{def:prav}.
\begin{figure}
    \centering
    \includegraphics[width=0.5\textwidth]{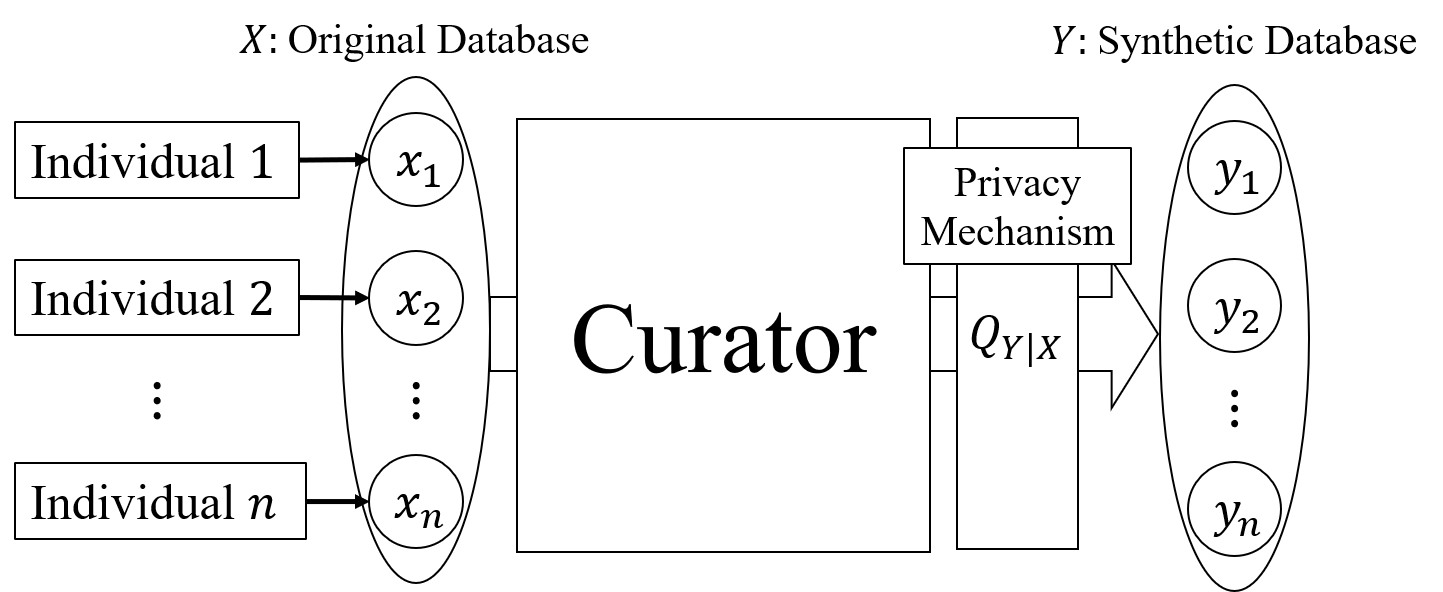}
    \caption{Global Privacy Framework: the curator collects all the data directly from individuals and publish them in a privacy-protected synthetic database}
    \label{fig:global}
\end{figure}


\begin{lem}\label{lem:general case bound}
Given a source set $\mathcal{P}$, let $D$ be a positive integer not greater than $n$, we have the following properties of privacy-distortion function: 

(1) For any $P\in\mathcal{P}$, $\tilde{\epsilon}^{*}(\mathcal{P},D)\geq\tilde{\epsilon}^{*}(P,D).$

(2) Let $Q_U$ be the uniform mechanism, i.e., $Q_U(y|x)=\frac{1}{m^n}$ for any $x,y\in\mathcal{X}.$ Then for $\frac{n(m-1)}{m}\leq D\leq 1$, $\tilde{\epsilon}^{*}(\mathcal{P},D)=\tilde{\epsilon}^{*}(Q_U)=0.$

(3) Let $Q_D$ be the mechanism given by Wang et al. \cite{Wang2016}, defined as $$Q_D(y|x)=\left(1-\frac{D}{n}\right)^n\left(\frac{(m-1)(n-D)}{D}\right)^{-d(x,y)},$$ then $\tilde{\epsilon}^{*}(\mathcal{P},D)\leq\tilde{\epsilon}(Q_D).$

\end{lem}

\begin{proof}[Proof of Lemma \ref{lem:general case bound}]
By Definition \ref{def:pri-dis}, for any mechanism $Q\in\mathcal{Q}^{*}(\mathcal{P},D)$, $Q$ is $(\mathcal{P},D)$-valid. Thus, $Q$ is $(P,D)$-valid for any $P\in\mathcal{P}$. Hence, $\tilde{\epsilon}^{*}(Q)=\tilde{\epsilon}^{*}(\mathcal{P},D)\geq\tilde{\epsilon}^{*}(P,D).$ 

Conversely, for any $(\mathcal{P},D)$-valid mechanism $Q$, the definition of privacy distortion function implies $\tilde{\epsilon}^{*}(\mathcal{P},D)\leq\tilde{\epsilon}^{*}(Q).$ Besides, it is easy to verify that $Q_U$ is $(\mathcal{P},D)$-valid for $\frac{n(m-1)}{m}\leq D\leq 1$ and $Q_D$ is $(\mathcal{P},D)$-valid for $0< D\leq 1$. This completes the proof.
\end{proof}


\begin{thm}\label{thm:tn,c2}
Let $\mathcal{P}$ be a source set such that each probability distribution in  $\mathcal{P}$ has full support and let $\theta^{*}:=m^{n} \cdot \max_{P\in\mathcal{P}}\min_{x\in\mathcal{X}}P(x)$. Then we have the following results on the privacy-distortion function with respect to the privacy notions in Definition \ref{def:prav}.

(1) (Differential privacy)
\beqs
&&\max\left\{0,\log(m-1)\frac{\theta^{*}n-D}{D}\right\}\\
&\leq& \tilde\epsilon_{DP}^{*}(\mathcal{P},D)\\
&\leq& \max\left\{0,\log\frac{(m-1)(n-D)}{D}\right\}.
\eeqs

(2) (Approximate differential privacy)
\beqs
&&\max\left\{0,\log(m-1)\frac{\theta^{*}n(1-\delta m^{n-1})-D}{D}\right\}\\
&\leq& \tilde\epsilon_{\delta}^{*}(\mathcal{P},D)\\
&\leq& \max\left\{0,\log(m-1)\frac{n-D}{D}(1-\delta(1-\frac{D}{n})^{-n})\right\}.
\eeqs

(3) (Maximal information)
\beqs
&&\max\left\{0,\log m(1-\frac{D}{\theta^{*}n})\right\}\\
&\leq& \tilde\epsilon_{MI}^{*}(\mathcal{P},D)\\
&\leq& \max\left\{0,-\log\min_{P_X\in\mathcal{P},y\in\mathcal{Y}}\sum_{l=0}^{m}\left((m-1)\frac{n-D}{D}\right)^{-l}P_{X}(\mathcal{N}_{l}(y))\right\}.
\eeqs

(4) (Maximal leakage)
\beqs
&&\max\left\{0,\log m(1-\frac{D}{\theta^{*}n})\right\}\\
&\leq& \tilde\epsilon_{ML}^{*}(\mathcal{P},D)\\
&\leq& \max\left\{0,n\log m\left(1-\frac{D}{n}\right)\right\}
\eeqs

(5) (R\'{e}nyi differential privacy)
\beqs
&&\max\left\{0,\log\frac{(m-1)(n-D/\theta^{*})^{\alpha/(\alpha-1)}}{D/\theta^{*}}-\frac{1}{\alpha-1}\log nm^{n-1}\right\}\\
&\leq& \tilde\epsilon_{\alpha,DP}^{*}(\mathcal{P},D)\\
&\leq& \max\Big\{0,\frac{1}{\alpha-1}\log\frac{D}{n(m-1)}\Big(m-2+\\
&&\left(\frac{(n-D)(m-1)}{D}\right)^{\alpha}+\left(\frac{(n-D)(m-1)}{D}\right)^{1-\alpha}\Big)\Big\}.
\eeqs

(6) (Sibson mutual information)
\beqs
&&\max\left\{0,\frac{\alpha-n}{\alpha-1}\log m+\frac{\alpha}{\alpha-1}\log\frac{1-D/n}{\frac{1}{2}(1-\theta^{*})(n+1)m+\theta^{*}}\right\}\\
&\leq& \tilde\epsilon_{\alpha,M}^{*}(\mathcal{P},D)\\
&\leq& \max\Big\{0,n\log m((n-D)^{\alpha}+D^{\alpha}(m-1)^{1-\alpha})^{1/(\alpha-1)}\\
&&-\frac{\alpha}{\alpha-1}n\log n\Big\}.
\eeqs

(7) (Mutual information)
\beqs
&&\max\Big\{0,\theta^{*}n\log \eta^{-\frac{1}{n}}(1-\frac{D}{n\theta^{*}})\\&&+\theta^{*}D\log\left(\frac{D}{(m-1)(n\theta^{*}-D)}\right)\Big\}\\
&\leq& \tilde\epsilon_{M}^{*}(\mathcal{P},D)\\
&\leq& \max\left\{0,n\log m(1-\frac{D}{n})+D\log\left(\frac{D}{(m-1)(n-D)}\right)\right\}
\eeqs
where $\eta=\max_{x\in\mathcal{X}}P^{\prime}(x)$ and $P^{\prime}=\arg\max_{P\in\mathcal{P}}\min_{x\in\mathcal{X}}P(x).$

In particular, $\tilde{\epsilon}^{*}(\mathcal{P},D)=0$ if $n(1-1/m)\leq D\leq 1$ for all notions aforementioned.
\end{thm}

\begin{proof}[Proof of Theorem \ref{thm:tn,c2}]
Based on  Lemma \ref{lem:general case bound}, we have already known $Q_U$ is $(\mathcal{P},D)$-valid for $\frac{n(m-1)}{m}\leq D\leq 1$. Thus $\tilde{\epsilon}^{*}(\mathcal{P},D)=0$ if $\frac{n(m-1)}{m}\leq D\leq 1$ for all notions aforementioned. 
Otherwise, for $0<D<\frac{n(m-1)}{m}$ we consider the upper- and lower-bounds, respectively.

(1) First, let us consider the upper bound on the privacy-distortion functions. By Lemma \ref{lem:general case bound}, $\tilde{\epsilon}^{*}(\mathcal{P},D)$ is upper-bounded by $\tilde{\epsilon}^{*}(Q_D)$. Hence, it suffices to calculate or upper-bound $\tilde{\epsilon}^{*}(Q_D)$. For simplicity, we denote $\frac{(m-1)(n-D)}{D}$ and $\left(1-\frac{D}{n}\right)^{n}$ by $A$ and $B$, respectively. Note that if $0<D<n(m-1)/m$ then $m^{-n}<B<1<A$. 
Let us start with $\alpha$-type notions, i.e., $\tilde{\epsilon}_{\alpha,DP}(\mathcal{P},D)$ and $\tilde{\epsilon}_{\alpha,M}(\mathcal{P},D)$.

1.1) For $\tilde{\epsilon}_{\alpha,DP}(\mathcal{P},D)$,  we have 
\beqs
\lefteqn{\tilde{\epsilon}_{\alpha,DP}(Q_D)=\max_{d(x,x^{\prime})=1}\frac{1}{\alpha-1}\log\sum_{y}Q_D^{\alpha}(y|x)Q^{1-\alpha}(y|x^{\prime})}\\
&&=\max_{d(x,x^{\prime})=1}\frac{1}{\alpha-1}\log\Big(\sum_{d(x,y)=d(x^{\prime},y)}A^{-d(x,y)}B\\&&\ \ +\sum_{d(x,y)=d(x^{\prime},y)+1}A^{-d(x,y)}B\cdot A^{1-\alpha}\\&&\ \ +\sum_{d(x,y)=d(x^{\prime},y)-1}A^{-d(x,y)}B\cdot A^{\alpha-1}\Big)\\
&&=\max_{d(x,x^{\prime})=1}\frac{1}{\alpha-1}\log\\&&\ \ \ \Big(\sum_{l=1}^{n}(m-2)\binom{m-1}{l-1}(m-1)^{l-1}A^{-l}B\\&&\ \ +\sum_{l=1}^{n}\binom{m-1}{l-1}(m-1)^{l-1}A^{-l}B\cdot A^{1-\alpha}\\&&\ \ +\sum_{l=0}^{n-1}A^{-l}B\cdot A^{\alpha-1}\Big)\\
&&=\max_{d(x,x^{\prime})=1}\frac{1}{\alpha-1}\log\sum_{l=0}^{n-1}(m-1)^{l}A^{-l}\Big((m-2)A^{-1}B\\&&\ \ +A^{-\alpha}B+A^{\alpha-1}B\Big)\\
&&=\frac{1}{\alpha-1}\log \frac{B}{A}\left((m-2)+A^{1-\alpha}+A^{\alpha}\right)(1+(m-1)A^{-1})^{n-1}\\
&&=\frac{1}{\alpha-1}\log\frac{D}{n(m-1)}\Big(m-2+\left(\frac{(n-D)(m-1)}{D}\right)^{\alpha}\\&&\ \ +\left(\frac{(n-D)(m-1)}{D}\right)^{1-\alpha}\Big).
\eeqs

1.2) For $\tilde{\epsilon}_{\alpha,M}(\mathcal{P},D)$, by the definition of Sibson $\alpha$-mutual information, we have 
\begin{eqnarray}
\nonumber\lefteqn{\tilde{\epsilon}_{\alpha,M}(Q_D)=\max_{P\in\mathcal{P}}\frac{\alpha}{\alpha-1}\log\sum_{y}\left(\sum_{x}P(x)Q^{\alpha}(y|x)\right)^{\frac{1}{\alpha}}}\\
\label{eq:smi}&&=\frac{\alpha}{\alpha-1}\log B+\frac{\alpha}{\alpha-1}\log\max_{P\in\mathcal{P}}\sum_{y}\left(\sum_{x}P(x)A^{-\alpha(d(x,y))}\right)^{\frac{1}{\alpha}}.
\end{eqnarray}

We now show that if $\mathcal{P}$ contains the uniform distribution $P_U$, then $\tilde{\epsilon}_{\alpha,M}(Q_D)$ obtains the maximal value at $P=P_U$. Otherwise, $\tilde{\epsilon}_{\alpha,M}(Q_D)$ is upper-bounded by $$\frac{\alpha}{\alpha-1}\log B+\frac{\alpha}{\alpha-1}\log\sum_{y}\left(\sum_{x}P_U(x)A^{-\alpha(d(x,y))}\right)^{\frac{1}{\alpha}}.$$ To prove this claim, let us consider the following optimization problem.
\begin{eqnarray}
\nonumber\lefteqn{\min_{P}-\sum_{y}\left(\sum_{x}P(x)A^{-\alpha(d(x,y))}\right)^{\frac{1}{\alpha}}}\\\label{op:smi-upper}
&&\text{subject to}\\\nonumber
&&(c1)\text{ }P(x)\geq0,\\\nonumber
&&(c2)\text{ }\sum_{x}P(x)=1.\nonumber
\end{eqnarray}

It is obvious that the feasible region of \eqref{op:smi-upper} consists of all the probability distributions. By the convexity of Sibson mutual information (see Lemma \ref{lem:con_of_pri} in Appendix \ref{app:bpp}), this optimization problem is convex. To find the optimal value, we utilize the Karush-Kuhn-Tucker (KKT) conditions \cite{Boyd2004}. The KKT conditions for convex problem \eqref{op:smi-upper} are as follows.


a) For $x\in\mathcal{X}$, 
\beqs&&\nabla_{P(x)} L(P(x),\lambda,\alpha_{x})\\&=&\lambda-\sum_{y}\frac{1}{\alpha}\left(\sum_{x}P(x)A^{-\alpha d(x,y)}\right)^{\frac{1-\alpha}{\alpha}}+\alpha_{x}\\&=&0\eeqs
where $L(P(x),\lambda,\alpha_{x})$ is the Lagrangian defined as follows.
\beqs
L(P(x),\lambda,\alpha_{x})=&-\sum_{y}\left(\sum_{x}P(x)A^{-\alpha(d(x,y))}\right)^{\frac{1}{\alpha}}\\&+\lambda\left(\sum_{x}P(x)-1\right)-\sum_{x}\alpha_{x}P(x).
\eeqs

b) $\sum_{x}P(x)=1$

c) $\alpha_{x}P(x)=0$, $\alpha_x\geq 0$, $P(x)\geq 0.$

It is easy to verify that 
$$(P=P_U,\ \lambda=\frac{1}{\alpha}m^{\frac{n(\alpha-1)}{\alpha}}\left(1+(m-1)A^{-\alpha}\right)^{\frac{n}{\alpha}},\ \alpha_x=0)$$
satisfies the above KKT conditions. Note that feasible solution of KKT conditions is also optimal for convex problem \cite{Boyd2004}. Thus 
\beqs
&&\tilde{\epsilon}_{\alpha,M}(Q_D)\\
&\leq&\frac{\alpha}{\alpha-1}\log B+\frac{\alpha}{\alpha-1}\log\sum_{y}\left(\sum_{x}P_U(x)A^{-\alpha(d(x,y))}\right)^{\frac{1}{\alpha}}\\
&=&n\log m((n-D)^{\alpha}+D^{\alpha}(m-1)^{1-\alpha})^{1/(\alpha-1)}-\frac{\alpha}{\alpha-1}n\log n.
\eeqs

1.3) For $\tilde{\epsilon}_{DP}$,  let $c_1=\frac{D}{m-1}$ and $c_2=n-D$, then $0<c_1<n/m<c_2<n$ and we have
\beqs
\lefteqn{\tilde{\epsilon}_{DP}(Q_D)=\lim_{\alpha\to\infty}\tilde{\epsilon}_{\alpha,DP}(Q_D)}\\
&&=\lim_{\alpha\to\infty}\frac{c_1\left(\frac{c_2}{c_1}\right)^{\alpha}\log\frac{c_2}{c_1}+c_2\left(\frac{c_1}{c_2}\right)^{\alpha}\log\frac{c_1}{c_2}}{\frac{m-2}{m-1}D+c_1\left(\frac{c_2}{c_1}\right)^{\alpha}+c_2\left(\frac{c_1}{c_2}\right)^{\alpha}}\\
&&=\lim_{\alpha\to\infty}\frac{c_1\left(\frac{c_2}{c_1}\right)^{\alpha}\log\frac{c_2}{c_1}}{\frac{m-2}{m-1}D+c_1\left(\frac{c_2}{c_1}\right)^{\alpha}}\\
&&=\log\frac{c_2}{c_1}=\log\frac{(m-1)(n-D)}{D}.
\eeqs

1.4) For $\tilde{\epsilon}_{ML}(Q_D)$, by the relation between Sibson mutual information and maximal leakage (see Lemma \ref{remark:implication} in Appendix \ref{app:bpp}),  we have
\beqs
\lefteqn{\tilde{\epsilon}_{ML}(Q_D)=\lim_{\alpha\to\infty}\tilde{\epsilon}_{\alpha,M}(Q_D)}\\
&&\leq n\log m+n\lim_{\alpha\to\infty}\frac{n}{\alpha-1}\log\left(\left(\frac{c_2}{n}\right)^{\alpha}+(m-1)\left(\frac{c_1}{n}\right)^{\alpha}\right)\\
&&=n\log m+n\frac{c_{2}^{\alpha}\log (c_{2}/n)+(m-1)c_{1}^{\alpha}\log (c_{1}/n)}{c_{2}^{\alpha}+(m-1)c_{1}^{\alpha}}\\
&&=n\log m(c_{2}/n)=n\log m(1-\frac{D}{n}).
\eeqs

1.5) For $\tilde{\epsilon}_{M}(Q_D)$,  by the relation between Sibson mutual information and mutual information (see Lemma \ref{remark:implication} in Appendix \ref{app:bpp}), we have
\beqs
\lefteqn{\tilde{\epsilon}_{M}(Q_D)=\lim_{\alpha\to 1}\tilde{\epsilon}_{\alpha,M}(Q_D)}\\
&&\leq n\log m+n\lim_{\alpha\to 1}\frac{n}{\alpha-1}\log\left(\left(\frac{c_2}{n}\right)^{\alpha}+(m-1)\left(\frac{c_1}{n}\right)^{\alpha}\right)\\
&&=n\log m+n\frac{c_{2}^{\alpha}\log (c_{2}/n)+(m-1)c_{1}^{\alpha}\log (c_{1}/n)}{c_{2}^{\alpha}+(m-1)c_{1}^{\alpha}}\\
&&=n\log m+n\frac{c_2\log(c_2/n)+(m-1)c_{1}\log(c_1/n)}{c_2+(m-1)c_1}\\
&&=n\log m(1-\frac{D}{n})+D\log\frac{D}{(m-1)(n-D)}.
\eeqs

1.6) For $\tilde{\epsilon}_{\delta}(Q_D)$, by Definition \ref{def:prav},
$$\epsilon_{\delta}(Q_D)=\max\left\{0,{\max_{y\in\mathcal{Y}, d(x,x^{\prime})=1}\log\frac{Q_{D}(y|x)-\delta}{Q_{D}(y|x^{\prime})}}\right\}.$$
Note that $|d(x,y)-d(x^{\prime},y)|\leq d(x,x^{\prime})$. Hence,
$$\tilde{\epsilon}_{\delta}(Q_D)=\max\left\{0,\log(m-1)\frac{n-D}{D}(1-\delta\left(1-\frac{D}{n}\right)^{-n})\right\}.$$

1.7) For $\tilde{\epsilon}_{MI}(Q_D)$, by Definition \ref{def:prav},
\beqs
\lefteqn{\epsilon_{MI}(Q)=\max_{P\in\mathcal{P},x\in\mathcal{X},y\in\mathcal{Y}}\log\frac{Q_{D}(y|x)}{\sum_{x\in\mathcal{X}}P_{X}(x)Q_{D}(y|x)}}\\
&&=-\log\min_{P_X\in\mathcal{P},y\in\mathcal{Y}}\sum_{l=0}^{m}\left((m-1)\frac{n-D}{D}\right)^{-l}P_{X}(\mathcal{N}_{l}(y)).
\eeqs

(2) Second, let us consider the lower bounds on the privacy-distortion functions. For any mechanism $Q$ in $\mathcal{Q}(\mathcal{P},D)$, it is $(P,D)$-valid for any $P\in\mathcal{P}$. Thus, for $$P^{\prime}=\arg\max_{P\in\mathcal{P}}\min_{x\in\mathcal{X}}P(x)\text{ and }\theta^{*}=m^n\cdot \max_{P\in\mathcal{P}}\min_{x\in\mathcal{X}}P(x),$$
\beqs
&&\sum_{x,y}m^{-n}\theta^{*}Q(y|x)d(x,y)\\&\leq&\sum_{x,y}P^{\prime}(x)Q(y|x)d(x,y)\\&=&\mathbb{E}_{P^{\prime},Q}[d(X,Y)]\leq D.
\eeqs
Therefore, we have 
\begin{equation}\label{ineq}
\sum_{x,y}Q(y|x)d(x,y)\leq \frac{D m^n}{\theta^{*}}.
\end{equation}
We complete the proof by contradiction.

2.1) For $\tilde{\epsilon}_{\delta}$,  assume there exists a $Q\in\mathcal{Q}(\mathcal{P},D)$ such that $\tilde{\epsilon}_{\delta}(Q)<\epsilon:=\log(m-1)\frac{\theta^{*}n(1-m^{n-1}\delta)-D}{D}$. 
Given an arbitrary vector $y\in\mathcal{X}$, we can always find $x\in\mathcal{N}_{l-1}(y)$ and $x^{\prime}\in\mathcal{N}_l(y)$ for any positive integer $l$ no larger than $n$. Since $x$ and $x^{\prime}$ are  neighboring, we have
\beqs
&&Q(y|x)<e^{\epsilon}Q(y|x^{\prime})+\delta\\
&\Rightarrow& (n-l+1)(m-1)Q(y|x)\\&&<e^{\epsilon}\sum_{x^{\prime}\in\mathcal{N}_{1}(x)\bigcap\mathcal{N}_{l}(y)}Q(y|x^{\prime})+\delta(n-l+1)(m-1)\\
&\Rightarrow& (n-l+1)(m-1)\sum_{x\in\mathcal{N}_{l-1}(y)}Q(y|x)\\&&<e^{\epsilon}\sum_{x\in\mathcal{N}_{l-1}(y)}\sum_{x^{\prime}\in\mathcal{N}_{1}(x)\bigcap\mathcal{N}_{l}(y)}Q(y|x^{\prime})+l N_{l}\delta\\
&\Rightarrow& (n-l+1)(m-1)\sum_{x\in\mathcal{N}_{l-1}(y)}Q(y|x)\\&&<e^{\epsilon}l\sum_{x\in\mathcal{N}_{l}(y)}Q(y|x)+l N_{l}\delta\\
&\Rightarrow& n(m-1)\sum_{x\in\mathcal{N}_{l-1}(y)}Q(y|x)\\&&<e^{\epsilon}l\sum_{x\in\mathcal{N}_{l}(y)}Q(y|x)\\&&+(m-1)(l-1)\sum_{x\in\mathcal{N}_{l-1}(y)}Q(y|x)+l N_{l}\delta\\
&\Rightarrow& n(m-1)\sum_{l=1}^{n}\sum_{x\in\mathcal{N}_{l-1}(y)}Q(y|x)\\&&<e^{\epsilon}\sum_{l=1}^{n}l\sum_{x\in\mathcal{N}_{l}(y)}Q(y|x)\\&&+(m-1)\sum_{l=1}^{n}(l-1)\sum_{x\in\mathcal{N}_{l-1}(y)}Q(y|x)+n(m-1)m^{n-1}\delta\\
&\Rightarrow& n(m-1)\sum_{x}Q(y|x)\\&&<(e^{\epsilon}+(m-1))\sum_{x}Q(y|x)d(x,y)+n(m-1)m^{n-1}\delta\\
&\Rightarrow& n(m-1)m^n\\&&<(e^{\epsilon}+(m-1))\sum_{x,y}Q(y|x)d(x,y)+n(m-1)m^{2n-1}\delta\\
&\Rightarrow& \sum_{x,y}Q(y|x)d(x,y)>\frac{n(m-1)m^{n}(1-m^{n-1}\delta)}{e^{\epsilon}+(m-1)}=\frac{D m^{n}}{\theta^{*}}.
\eeqs
Contradiction! Thus $\tilde{\epsilon}^{*}_{\delta}(\mathcal{P},D)\geq \log(m-1)\frac{\theta^{*}n(1-m^{n-1}\delta)-D}{D}.$

2.2)  The $\epsilon$-DP can be regarded as $(\epsilon,0)$-DP. Thus, for $\tilde{\epsilon}_{DP}$, we have $\tilde{\epsilon}^{*}_{DP}(\mathcal{P},D)\geq \log\frac{(m-1)(\theta^{*}n-D)}{D}.$

2.3) For $\tilde{\epsilon}_{MI}$,  assume there exists a $Q\in\mathcal{Q}(\mathcal{P},D)$ such that $\tilde{\epsilon}_{MI}(Q)<\epsilon:=\log m(1-\frac{D}{\theta^{*}n})$,  then for any $x$, $y\in\mathcal{X}$ and $P\in\mathcal{P}$,
\beqs
&&Q(y|x)<e^{\epsilon}\sum_{z}P(z)Q(y|z)\\
&\Rightarrow& \sum_{x\in\mathcal{N}_{l}(y)}Q(y|x)<e^{\epsilon}N_{l}\sum_{z}P(z)Q(y|z)\\
&\Rightarrow& \sum_{y}\sum_{x\in\mathcal{N}_{l}(y)}Q(y|x)<e^{\epsilon}N_{l}\sum_{y,z}P(z)Q(y|z)\\
&\Rightarrow& \sum_{y}\sum_{x\in\mathcal{N}_{l}(y)}Q(y|x)<e^{\epsilon}N_{l}\\
&\Rightarrow& \sum_{l=0}^{n}(n-l)\sum_{y}\sum_{x\in\mathcal{N}_{l}(y)}Q(y|x)<e^{\epsilon}\sum_{l=0}^{n}(n-l)N_{l}\\
&\Rightarrow& n m^n-\sum_{x,y}Q(y|x)d(x,y)<e^{\epsilon}n m^n-e^{\epsilon}n m^{n}(1-\frac{1}{m})\\
&\Rightarrow& \sum_{x,y}Q(y|x)d(x,y)>n m^{n} (1-\frac{e^{\epsilon}}{m})=\frac{D m^n}{\theta^{*}}.
\eeqs
Contradiction! Thus $\tilde{\epsilon}^{*}_{MI}(\mathcal{P},D)\geq \log m(1-\frac{D}{\theta^{*}n}).$

2.4) For $\tilde{\epsilon}_{ML}$,  assume there exists a  $Q\in\mathcal{Q}(\mathcal{P},D)$ such that $\tilde{\epsilon}_{ML}(Q)<\epsilon:=\log m(1-\frac{D}{n\theta^{*}})$, then
\beqs
&&\sum_{y}\max_{x}Q(y|x)<e^{\epsilon}\\
&\Rightarrow& \sum_{y}\max_{x\in\mathcal{N}_{l}(y)}Q(y|x)<e^{\epsilon}\text{for }0\leq l\leq n\\
&\Rightarrow& \sum_{x,y}Q(y|x)d(x,y)>n m^{n} (1-\frac{e^{\epsilon}}{m})=\frac{D m^n}{\theta^{*}}.
\eeqs
Contradiction! Thus $\tilde{\epsilon}^{*}_{ML}(\mathcal{P},D)\geq \log m(1-\frac{D}{n\theta^{*}}).$

2.5) For $\tilde{\epsilon}_{\alpha,DP}$,  assume there exists a $Q\in\mathcal{Q}(\mathcal{P},D)$ such that $\tilde{\epsilon}_{\alpha,DP}(Q)<\epsilon:=\log\frac{(m-1)(n-D/\theta^{*})^{\frac{\alpha}{\alpha-1}}}{D/\theta^{*}}-\frac{1}{\alpha-1}\log n m^{n-1}$. For any vector $y\in\mathcal{X}$, we can always find  $x\in\mathcal{N}_{l-1}(y)$ and $x^{\prime}\in\mathcal{N}_l(y)$ for any positive integer $l$ no larger than $n$. Since $x$ and $x^{\prime}$ are two neighbors, then by the probability preservation property of R\'{e}nyi divergence (see Lemma \ref{lem:renyi} in Appendix \ref{app:bpp}), we have
\beqs
&&Q(y|x)<\left(e^{\epsilon}Q(y|x^{\prime})\right)^{\frac{\alpha-1}{\alpha}}\\
&\Rightarrow& (n-l+1)(m-1)Q(y|x)\\&&<\sum_{x^{\prime}\in\mathcal{N}_{1}(x)\bigcap\mathcal{N}_{l}(y)}\left(e^{\epsilon}Q(y|x^{\prime})\right)^{\frac{\alpha-1}{\alpha}}\\
&\Rightarrow& (n-l+1)(m-1)\sum_{x\in\mathcal{N}_{l-1}(y)}Q(y|x)\\&&<l\sum_{x\in\mathcal{N}_{l}(y)}\left(e^{\epsilon}Q(y|x)\right)^{\frac{\alpha-1}{\alpha}}\\
&\Rightarrow& n(m-1)\sum_{x\in\mathcal{N}_{l-1}(y)}Q(y|x)\\&&<l\sum_{x\in\mathcal{N}_{l}(y)}\left(e^{\epsilon}Q(y|x)\right)^{\frac{\alpha-1}{\alpha}}\\&&+(m-1)(l-1)\sum_{x\in\mathcal{N}_{l-1}(y)}Q(y|x)\\
&\Rightarrow& n(m-1)\sum_{l=1}^{n}\sum_{x\in\mathcal{N}_{l-1}(y)}Q(y|x)\\&&<\sum_{l=1}^{n}l\sum_{x\in\mathcal{N}_{l}(y)}\left(e^{\epsilon}Q(y|x)\right)^{\frac{\alpha-1}{\alpha}}\\&&+(m-1)\sum_{l=1}^{n}(l-1)\sum_{x\in\mathcal{N}_{l-1}(y)}Q(y|x)\\
&\Rightarrow& n(m-1)\sum_{x}Q(y|x)\\&&<\sum_{x}\left(e^{\epsilon}Q(y|x)\right)^{\frac{\alpha-1}{\alpha}}d(x,y)\\&&+(m-1)\sum_{x}Q(y|x)d(x,y)\\
&\Rightarrow& n(m-1)m^n\\&&<\sum_{x,y}\left(e^{\epsilon}Q(y|x)\right)^{\frac{\alpha-1}{\alpha}}d(x,y)\\&&+(m-1)\sum_{x,y}Q(y|x)d(x,y).
\eeqs
By H\"{o}lder inequality, 
\beqs
n(m-1)m^n<\left(e^{\epsilon}\sum_{x,y}Q(y|x)d(x,y)\right)^{\frac{\alpha-1}{\alpha}}\left(\sum_{x,y}d(x,y)\right)^{\frac{1}{\alpha}}\\+(m-1)\sum_{x,y}Q(y|x)d(x,y)\\
=\left(n(m-1)m^{2n-1}\right)^{\frac{1}{\alpha}}\left(e^{\epsilon}\sum_{x,y}Q(y|x)d(x,y)\right)^{\frac{\alpha-1}{\alpha}}\\+(m-1)\sum_{x,y}Q(y|x)d(x,y).
\eeqs
However, the inequality \eqref{ineq} implies that 
\beqs
&&\left(n(m-1)m^{2n-1}\right)^{\frac{1}{\alpha}}\left(e^{\epsilon}\sum_{x,y}Q(y|x)d(x,y)\right)^{\frac{\alpha-1}{\alpha}}\\&&+(m-1)\sum_{x,y}Q(y|x)d(x,y)\\
&&\leq \left(n(m-1)m^{2n-1}\right)^{\frac{1}{\alpha}}\left(\frac{e^{\epsilon}D m^n}{\theta^{*}}\right)^{\frac{\alpha-1}{\alpha}}+(m-1)\frac{D m^n}{\theta^{*}}\\
&&=n(m-1)m^n.
\eeqs
Contradiction! Thus $$\tilde{\epsilon}^{*}_{\alpha,DP}(\mathcal{P},D)\geq \log\frac{(m-1)(n-D)^{\frac{\alpha}{\alpha-1}}}{D}-\frac{1}{\alpha-1}\log n m^{n-1}.$$

2.6) For $\tilde{\epsilon}_{\alpha,M}$,  assume there exists a $Q\in\mathcal{Q}(\mathcal{P},D)$ such that $\tilde{\epsilon}_{\alpha,M}(Q)<\epsilon:=\frac{\alpha-n}{\alpha-1}\log m+\frac{\alpha}{\alpha-1}\log\frac{1-D/n}{\frac{1}{2}(1-\theta^{*})(n+1)m+\theta^{*}}$. By the probability preservation property of Sibson mutual information (for fixed $x\in\mathcal{X}$, take $E^{(l)}=\{(x,y)|d(x,y)=l\}$ into Lemma \ref{lem:sibson} in Appendix \ref{app:bpp}), we have for any $P\in\mathcal{P}$,
\beqs
&& \sum_{d(x,y)=l}P(x)Q(y|x)\\
&\leq& \max_{y}\left(P(\mathcal{N}_{l}(y))e^{I_{\alpha}(X;Q(X))}\right)^{\frac{\alpha-1}{\alpha}}\\
&<& \max_{y}\left(P(\mathcal{N}_{l}(y))e^{\epsilon}\right)^{\frac{\alpha-1}{\alpha}}\\
&\leq& \left(e^{\epsilon}\left(1-\frac{\theta^{*}(m^n-N_{l})}{m^n}\right)\right)^{\frac{\alpha-1}{\alpha}}.
\eeqs
Thus, 
\beqs
&&\sum_{l=0}^{n}(n-l)\sum_{\mathclap{d(x,y)=l}}P(x)Q(y|x)\\&&<\left(\frac{e^{\epsilon}}{m^n}\right)^{\frac{\alpha-1}{\alpha}}\sum_{l=0}^{n}(n-l)\left((1-\theta^{*})m^n+\theta^{*}N_l\right)^{\frac{\alpha-1}{\alpha}}\\
&\Rightarrow& n-\sum_{x,y}P(x)Q(y|x)d(x,y)\\&&<\left(\frac{e^{\epsilon}}{m^n}\right)^{\frac{\alpha-1}{\alpha}}\left((1-\theta^{*})m^n\binom{n+1}{2}+\theta^{*}n m^{n-1}\right)\\
&\Rightarrow& \sum_{x,y}P(x)Q(y|x)d(x,y)>D.
\eeqs
Contradiction! Thus $$\tilde{\epsilon}^{*}_{\alpha,M}(\mathcal{P},D)\geq \frac{\alpha-n}{\alpha-1}\log m+\frac{\alpha}{\alpha-1}\log\frac{1-D/n}{\frac{1}{2}(1-\theta^{*})(n+1)m+\theta^{*}}.$$

2.7) We restrict the input distribution to $$P^{\prime}=\arg\max_{P\in\mathcal{P}}\min_{x\in\mathcal{X}}P(x).$$ By Lemma \ref{lem:general case bound}, $\tilde{\epsilon}^{*}(\mathcal{P},D)\geq\tilde{\epsilon}^{*}(P^{\prime},D)$. To evaluate $\tilde{\epsilon}^{*}(P^{\prime},D)$, let us consider the following convex optimization problem.
\begin{eqnarray}
\nonumber\lefteqn{\min_{Q(y|x)}\sum_{x,y}P^{\prime}(x)Q(y|x)\log\frac{Q(y|x)}{\sum_{x}P_U(x)Q(y|x)}}\\
\label{op:mi-lower}&&\text{subject to}\\
\nonumber&&(c1)\text{ }\sum_{x,y}P^{\prime}(x)Q(y|x)d(x,y)\leq D,\\
\nonumber&&(c2)\text{ }\sum_{y}Q(y|x)=1,\\
\nonumber&&(c3)\text{ }Q(y|x)\geq0.
\end{eqnarray}
All the conditions in \eqref{op:mi-lower} imply feasible $Q$ is $(P^{\prime},D)$-valid and thus the optimal value is $\tilde{\epsilon}^{*}(P^{\prime},D)$. It is difficult to solve \eqref{op:mi-lower} directly. Instead, let us consider the following relaxed optimization problem.
\begin{eqnarray}
\nonumber\lefteqn{\min_{Q(y|x)}m^{-n}\theta^{*}\sum_{x,y}Q(y|x)\log\frac{Q(y|x)}{\sum_{x}\eta Q(y|x)}}\\
\label{op:mi-lower-relax}&&\text{subject to}\\
\nonumber&&(c1)\text{ }\sum_{x,y}Q(y|x)d(x,y)\leq \frac{D m^n}{\theta^{*}},\\
\nonumber&&(c2)\text{ }\sum_{y}Q(y|x)=1,\\
\nonumber&&(c3)\text{ }Q(y|x)\geq0
\end{eqnarray}
where $\eta=\max_{x\in\mathcal{X}}P^{\prime}(x)$. Note that $(c1)$ in \eqref{op:mi-lower} implies $(c1)$ in \eqref{op:mi-lower-relax}. Thus, the optimal value of \eqref{op:mi-lower-relax} is not greater than $\tilde{\epsilon}^{*}(P^{\prime},D)$. By the convexity of mutual information (see Lemma \ref{lem:con_of_pri} in Appendix \ref{app:bpp}), optimization problem \eqref{op:mi-lower-relax} is convex. To find the optimal value, we utilize the  KKT conditions again. The KKT conditions for convex problem \eqref{op:mi-lower-relax} are as follows.

a) For any $x$, $y\in\mathcal{X}$, $$\nabla_{Q} L(Q,\lambda,\alpha_x,\beta_{xy})=0$$ where 
$L(Q,\lambda,\alpha_x,\beta_{xy})$ is the Lagrangian of \eqref{op:mi-lower-relax} defined as 
\beqs
L(Q,\lambda,\alpha_x,\beta_{xy}):=m^{-n}\theta^{*}\sum_{x,y}Q(y|x)\log\frac{Q(y|x)}{\sum_{x}\eta Q(y|x)}\\+\lambda(\sum_{x,y}Q(y|x)d(x,y)-\frac{D m^n}{\theta^{*}})\\+\sum_{x}\alpha_x\left(\sum_{y}Q(y|x)-1\right)-\sum_{x,y}\beta_{xy}Q(y|x).
\eeqs

b) $\lambda(\sum_{x,y}Q(y|x)d(x,y)-\frac{D m^n}{\theta^{*}})=0$, $\lambda\geq0$ and $$\sum_{x,y}Q(y|x)d(x,y)\leq \frac{D m^n}{\theta^{*}}.$$

c) For $x\in\mathcal{X}$, $\sum_{y}Q(y|x)=1.$

d) For any $x$, $y\in\mathcal{X}$, $\beta_{xy}Q(y|x)=0$, $\beta_{xy}\geq0$, $Q(y|x)\geq0.$

It is easy to verify that
$(Q=Q_{D/\theta^{*}}, \lambda=\frac{\theta^{*}}{m^n}\log\frac{(m-1)(n\theta^{*}-D)}{D}, \alpha_x= \frac{n\theta^{*}}{m^n}\log\frac{\eta^{1/n}n}{m(n-D/\theta^{*})}, \beta_{xy}=0)$
satisfies the KKT conditions.
Thus 
\beqs
\tilde{\epsilon}_{M}^{*}(\mathcal{P},D)&\geq&\tilde{\epsilon}_{M}^{*}(P^{\prime},D)\\
&\geq& m^{-n}\theta^{*}\sum_{x,y}Q_{D/\theta^{*}}(y|x)\log\frac{Q_{D/\theta^{*}}(y|x)}{\sum_{x}\eta Q_{D/\theta^{*}}(y|x)}\\
&=& \theta^{*}n\log \eta^{-\frac{1}{n}}(1-\frac{D}{n\theta^{*}})\\&&+\theta^{*}D\log\left(\frac{D}{(m-1)(n\theta^{*}-D)}\right).
\eeqs
\end{proof}

It is seen in Theorem \ref{thm:tn,c2} that the difference between lower- and upper-bounds is decreasing in $\theta^{*}$. In fact, the difference vanishes for some special cases. Assume that the source set $\mathcal{P}$ in Theorem \ref{thm:tn,c2} contains the uniform distribution $P_U$. Then $\arg\max_{P\in\mathcal{P}}\min_{x\in\mathcal{X}}P(x)=P_U$, which implies $\theta^{*}=1$ and $\eta=m^{-n}$. Hence, we have the following results on privacy-distortion functions.

\begin{thm}\label{thm:tn,c1}
Let $\mathcal{P}$ be a source set over $\mathcal{X}=\{1,...,m\}^n$ in Class I, then we have the following results 
on the privacy-distortion function with respect to the privacy notions in Definition \ref{def:prav},

(1) (Differential privacy)
$$\tilde{\epsilon}^{*}_{DP}(\mathcal{P},D)=\max\left\{0,\log\frac{(m-1)(n-D)}{D}\right\}.
$$

(2) (Approximate differential privacy)
\begin{eqnarray*}
&&\max\left\{0,\log(m-1)\frac{n(1-\delta m^{n-1})-D}{D}\right\}\\
&\leq&\tilde{\epsilon}^{*}_{\delta}(\mathcal{P},D)\\
&\leq&\max\left\{0,\log(m-1)\frac{n-D}{D}(1-\delta(1-\frac{D}{n})^{-n})\right\}.
\end{eqnarray*}

(3) (Maximal information)
$$\max\left\{0,\log m(1-\frac{D}{n})\right\}\leq\tilde{\epsilon}^{*}_{MI}(\mathcal{P},D)\leq\max\left\{0,\tilde{\epsilon}_{MI}(Q_D)\right\},$$
where $$\tilde{\epsilon}_{MI}(Q_D)=-\log\min_{P_X\in\mathcal{P},y\in\mathcal{Y}}\sum_{l=0}^{m}\left((m-1)\frac{n-D}{D}\right)^{-l}P_{X}(\mathcal{N}_{l}(y)).$$

(4) (Maximal leakage)
\begin{eqnarray*}
&&\max\left\{0,\log m\left(1-\frac{D}{n}\right)\right\}\\
&\leq&\tilde{\epsilon}^{*}_{ML}(\mathcal{P},D)\\
&\leq& \max\left\{0,n\log m\left(1-\frac{D}{n}\right)\right\}.
\end{eqnarray*}

(5) (R\'{e}nyi differential privacy)
\beqs
&&\max\left\{0,\log\frac{(m-1)(n-D)^{\alpha/(\alpha-1)}}{D}-\frac{1}{\alpha-1}\log nm^{n-1}\right\}\\
&\leq& \tilde{\epsilon}^{*}_{\alpha,DP}(\mathcal{P},D) \\
&\leq&
\max\Big\{0,\frac{1}{\alpha-1}\log\frac{D}{n(m-1)}\Big(m-2+\\
&&\left(\frac{(n-D)(m-1)}{D}\right)^{\alpha}+\left(\frac{(n-D)(m-1)}{D}\right)^{1-\alpha}\Big)\Big\}.
\eeqs

(6) (Sibson mutual information)
\beqs
&&\max\left\{0,\tilde{\epsilon}^{*}_{M}(\mathcal{P},D),\frac{1}{\alpha-1}\log m^{\alpha-n}(1-\frac{D}{n})^{\alpha}\right\}\\
&\leq&\tilde{\epsilon}^{*}_{\alpha,M}(\mathcal{P},D)\\
&\leq&\max\left\{0,\tilde{\epsilon}_{\alpha,M}(Q_D)\right\}
\eeqs

where $\tilde{\epsilon}_{\alpha,M}(Q_D)=n\log m((n-D)^{\alpha}+D^{\alpha}(m-1)^{1-\alpha})^{1/(\alpha-1)}-\frac{\alpha}{\alpha-1}n\log n.$

(7) (Mutual information)
$$\tilde{\epsilon}^{*}_{M}(\mathcal{P},D)=\max\left\{0,n\log m(1-\frac{D}{n})+D\log\left(\frac{D}{(m-1)(n-D)}\right)\right\}.$$

In particular, $\tilde{\epsilon}^{*}(\mathcal{P},D)=0$ if $\frac{n(m-1)}{m}\leq D\leq 1$ for all notions aforementioned.
\end{thm}

\section{Local Privacy Framework}\label{sec:4}
In this section, we assume $|\mathcal{X}|=m>0$ and consider the local privacy framework, i.e., the privacy mechanism is applied before data uploading (see Fig. \ref{fig:local}). 
The distance metric $d(x,y)$ is actually the discrete distance, i.e., $$d(x,y)=
\begin{cases}
0,&\text{ }x=y,\\
1,&\text{ }x\neq y.
\end{cases}
$$
Hence, for any $x\neq x'$, $x$ and $x'$ are neighborhood.

Given  $0<D\leq 1$,   $Q$ belongs to $\mathcal{Q}(P,D)$ if and only if $\sum_{i=1}^{m}P_{i}Q(i|i)\geq1-D.$ We claim that a good choice of output space is the synthetic one i.e., $\mathcal{Y}=\mathcal{X}$ (See Lemma \ref{lem:size_of_out} in Appendix \ref{app:st}).

\begin{figure}
    \centering
    \includegraphics[width=0.5\textwidth]{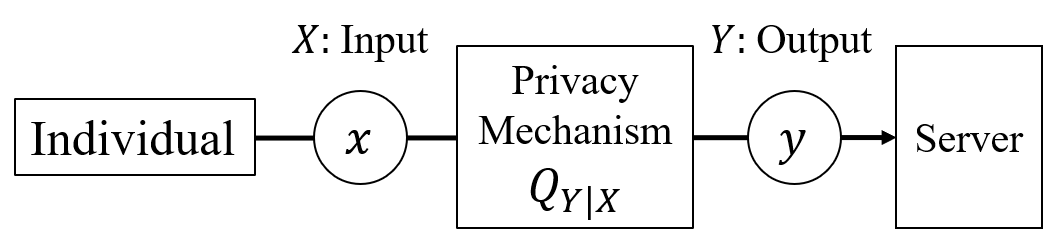}
    \caption{Local Privacy Framework: the individuals do not trust the data server, so the privacy mechanism is applied before data uploading}
    \label{fig:local}
\end{figure}

\subsection{Source Set of Class I}
Let us start with the source set of Class I, for which we take $n=1$ into Theorem \ref{thm:tn,c1} and get the following results directly.

\begin{prop}\label{thm:t1,c1}
Let $\mathcal{P}$ be a source set over $\mathcal{X}=\{1,...,m\}$ in Class I, then we have the following results 
on the privacy-distortion function with respect to the privacy notions in Definition \ref{def:prav},

(1) (Differential privacy)
$$\epsilon^{*}_{DP}(\mathcal{P},D)=
\begin{cases}
\log\frac{(m-1)(1-D)}{D}, & 0<D <\frac{m-1}{m},\\
0, & \frac{m-1}{m} \leq D \leq 1;
\end{cases}$$

(2) (Approximate differential privacy)
$$\epsilon^{*}_{\delta}(\mathcal{P},D)=
\begin{cases}
\log\frac{(m-1)(1-D-\delta)}{D}, & 0<D <(1-\delta)\frac{m-1}{m},\\
0, & (1-\delta)\frac{m-1}{m} \leq D \leq 1;
\end{cases}$$

(3) (Maximal information)
$$
\begin{cases}
\log m(1-D)\leq\epsilon^{*}_{MI}(\mathcal{P},D)\leq\epsilon_{MI}(Q_D), & 0<D <\frac{m-1}{m},\\
\epsilon^{*}_{MI}(\mathcal{P},D)=0, & \frac{m-1}{m} \leq D \leq 1
\end{cases}$$
where $\epsilon_{MI}(Q_D)=\log\frac{1-D}{P^{*}_{\mathcal{P}}(1-D-\frac{D}{m-1})+\frac{D}{m-1}}$ and $P^{*}_{\mathcal{P}}=\min_{1\leq i\leq m}$ $\{P_{i}|P\in\mathcal{P}\};$

(4)  (Maximal leakage)
$$\epsilon^{*}_{ML}(\mathcal{P},D)=
\begin{cases}
\log m(1-D), & 0<D <\frac{m-1}{m},\\
0, & \frac{m-1}{m} \leq D \leq 1;
\end{cases}$$

(5) (R\'{e}nyi differential privacy)
\beqs
&&\max\{0,\log \frac{(m-1)(1-D)^{\alpha/(\alpha-1)}}{D}\}\\
&\leq& \epsilon^{*}_{\alpha,DP}(\mathcal{P},D)\\
&\leq& \max\{0,\epsilon_{\alpha,DP}(Q_D)\},
\eeqs
where \beqs
\epsilon_{\alpha,DP}(Q_D)=\frac{1}{\alpha-1}\log\frac{D}{m-1}\Big(m-2+\left(\frac{(1-D)(m-1)}{D}\right)^{\alpha}&\\
+\left(\frac{(1-D)(m-1)}{D}\right)^{1-\alpha}\Big)&;
\eeqs

(6) (Sibson mutual information)
\beqs
&&\max\left\{\epsilon^{*}_{M}(\mathcal{P},D),\log m(1-D)^{\frac{\alpha}{\alpha-1}}\right\}\\
&\leq&\epsilon^{*}_{\alpha,M}(\mathcal{P},D)\\&\leq&\max\left\{0,\epsilon_{\alpha,M}(Q_D)\right\},
\eeqs
where $\epsilon_{\alpha,M}(Q_D)=\log m((1-D)^{\alpha}+D^{\alpha}(m-1)^{1-\alpha})^{1/(\alpha-1)};$

(7) (Mutual information)
$$\epsilon^{*}_{M}(\mathcal{P},D)=
\begin{cases}
\log m(1-D)\left(\frac{D}{(m-1)(1-D)}\right)^{D}, & 0<D <\frac{m-1}{m},\\
0, & \frac{m-1}{m} \leq D \leq 1.
\end{cases}$$

\end{prop}


\subsection{Source Set of Class II}
By taking $n=1$ into Theorem \ref{thm:tn,c2}, one can obtain the properties of privacy-distortion functions for any source set $\mathcal{P}$. However, if $\mathcal{P}$ does not contain the uniform distribution $P_U$ then Theorem \ref{thm:tn,c2} fails to give the lower bounds sometimes. Thus, in this subsection, we introduce a different method to address the PUT problem for the source set of Class II, and we mainly study the differential privacy-, approximate differential privacy- and maximal leakage-distortion trade-off.
First, we consider the scenario in which the distribution on $\mathcal{X}$ is given. 
\begin{thm}\label{thm:DP2}
Let $P$ be a distribution with full support on $\mathcal{X}$ such that $0\leq P_{m}\leq P_{m-1}\leq\cdots\leq P_{1}\leq1.$
Let $D^{(k)}=\sum_{j=m-k+1}^{m}P_{j}$ for $0\leq k\leq m.$ Then
$$\epsilon_{\delta}^{*}(P,D)=\max\left\{0,\min_{D>(1-\delta)D^{(k-1)}}\log\frac{(m-k)(1-D-\delta)}{D-(1-\delta)D^{(k-1)}}\right\}.$$
\end{thm}

Before proving Theorem \ref{thm:DP2}, we need introduce the following useful lemmas.

\begin{lem}\label{lem:02}
For $0<\delta<1$, $\epsilon^{*}_{\delta}(P,D)=0$ if and only if $D\geq(1-\delta)D^{(m-1)}.$
\end{lem}

\begin{proof}[Proof of Lemma \ref{lem:02}]
On one hand, if $\epsilon^{*}_{\delta}(P,D)=0$, then there exists a mechanism $Q$ such that $\epsilon_{\delta}(Q)=0$ and $\mathbb{E}_{P,Q}[d(x,y)]\leq D$. Thus, $|Q(j|i)-Q(j|k)|\leq\delta$ for any $i$, $j$ and $k$ in $\mathcal{X}$. Hence,
\beqs
\lefteqn{\mathbb{E}_{P,Q}[d(x,y)]=\sum_{j=1}^{m}P_{j}(1-Q(j|j))}\\
&&=P_{1}(1-Q(1|1))+\sum_{j=2}^{m}P_{j}(1-Q(j|j))\\
&&\geq P_{1}(1-Q(1|1))+\sum_{j=2}^{m}P_{j}(1-Q(j|1)-\delta)\\
&&=\sum_{j=1}^{m}P_{j}(1-Q(j|1)-\delta)+\delta P_{1}\\
&&=1-\delta+\delta P_{1}-\sum_{j=1}^{m}P_{j}Q(j|1)\\
&&\geq 1-\delta+\delta P_{1}-P_{1}=(1-\delta)D^{(m-1)}.
\eeqs
Thus $(1-\delta)D^{(m-1)}\leq\mathbb{E}_{P,Q}[d(x,y)]\leq D.$

On the other hand, if  $D\geq (1-\delta)D^{(m-1)}$, we define mechanism $Q_{\delta}$ as follows.
\beqs
Q_{\delta} 
&=&\begin{pmatrix}  
Q_{\delta}(1|1) & Q_{\delta}(2|1) & \cdots & Q_{\delta}(m|1)\\  
Q_{\delta}(1|2) & Q_{\delta}(2|2) & \cdots & Q_{\delta}(m|2)\\
\vdots & \vdots & \ddots & \vdots\\
Q_{\delta}(1|m) & Q_{\delta}(2|m) & \cdots & Q_{\delta}(m|m)
\end{pmatrix} \\
&=&\begin{pmatrix}  
1 & 0 & 0 & \cdots & 0\\  
1-\delta & \delta & 0 & \cdots & 0\\
1-\delta & 0 & \delta & \cdots & 0\\
\vdots & \vdots & \ddots & \vdots\\
1-\delta & 0 & 0 & \cdots & \delta
\end{pmatrix}.
\eeqs
Note that $\epsilon_{\delta}(Q_{\delta})=0$ and $\mathbb{E}_{P,Q_{\delta}}[d(x,y)]=(1-\delta)D^{(m-1)}\leq D.$ Thus, $\epsilon^{*}_{\delta}(P,D)=0$.
\end{proof}

\begin{lem}\label{lem:eq2}
For $0<\delta<1$, if $\epsilon^{*}_{\delta}(P,D)>0$, then there exists a mechanism $Q\in\mathcal{Q}_{\delta}^{*}(P,D)$ such that 

(1) $\delta\leq Q(m|m)\leq Q(m-1|m-1)\leq\cdots\leq Q(1|1)\leq1;$

(2) $Q(j|j)=e^{\epsilon^{*}_{\delta}(P,D)}Q(j|1)+\delta$ for $2\leq j\leq m.$
\end{lem}

To prove Lemma \ref{lem:eq2}, we improve the coloring method introduced by Kalantari et al.\cite{Kalantari2018}.  The  details of the proof  are provided in Appendix \ref{app:cs}, \ref{app:st} and \ref{app:sm}.

\begin{proof}[Proof of Theorem \ref{thm:DP2}]
For $D\geq(1-\delta)D^{(m-1)}$, Lemma \ref{lem:02} proves the statement.

For $0<D<(1-\delta)D^{(m-1)}$, let $Q^{*}$ be a mechanism satisfying the conditions in Lemma \ref{lem:eq2}. Then $Q^{*}(j|1)$ $=$ $e^{-\epsilon_{\delta}^{*}(P,D)}(Q^{*}(j|j)-\delta)$ for $2\leq j\leq m$, which implies $$\epsilon_{\delta}^{*}(P,D)=\log\frac{\sum_{j=2}^{m}(Q^{*}(j|j)-\delta)}{1-Q^{*}(1|1)}.$$
Let  us consider the following optimization problem,
\begin{eqnarray}
\nonumber\lefteqn{d_1:=\min_{\alpha_1,\alpha_2,\cdots,\alpha_m}\log\frac{\sum_{j=2}^{m}(\alpha_j-\delta)}{1-\alpha_1}}\\
\nonumber&&\text{subject to}\\
\label{op:adp1}&&(c1)\text{ }\sum_{j=1}^{m}P_j\alpha_j\geq1-D,\\
\nonumber&&(c2)\text{ }\sum_{j=1}^{m}(\alpha_j-\delta)\geq1-\delta,\\
\nonumber&&(c3)\text{ }\alpha_{j}\leq\alpha_{j-1}\text{ for }1\leq j\leq m+1
\end{eqnarray}
where $\alpha_0=1$ and $\alpha_{m+1}=\delta$. Note that $\{Q^{*}(j|j)\}_{1\leq j\leq m}$ is a feasible solution of \eqref{op:adp1}, and hence $\epsilon_{\delta}^{*}(P,D)\geq d_1.$ Conversely, let $\{\alpha_{1}^{*},\alpha_{2}^{*},\cdots,\alpha_{m}^{*}\}$ be the optimal solution of \eqref{op:adp1}. We define the following mechanism.
$$
Q(j|i)=
\begin{cases}
\alpha_{j}^{*}&\text{ if }i=j,\\
\frac{(\alpha_{j}^{*}-\delta)(1-\alpha_{i}^{*})}{\sum_{j\neq i}(\alpha_{j}^{*}-\delta)}&\text{ if }i\neq j.
\end{cases}
$$ 
For any $j\in[1,m]$, $\sum_{i=1}^{m}Q(j|i)=1$, that is,  $Q$ is a conditional probability. The first condition in \eqref{op:adp1} implies that $Q$ is within $\mathcal{Q}(P,D)$. The other two imply that $$\epsilon_2(Q)=\max_{1\leq i\leq m}\frac{\sum_{j\neq i}(\alpha_{j}^{*}-\delta)}{1-\alpha_{i}^{*}}=\frac{\sum_{j\neq 1}(\alpha_{j}^{*}-\delta)}{1-\alpha_{1}^{*}}=d_1.$$ Thus, we have $d_1=\epsilon_{\delta}^{*}(P,D).$ Moreover, condition $(c1)$ in \eqref{op:adp1} combined with the assumption $P_j$ is non-increasing in $j$ yields that $$\alpha_{1}=\sum_{j=1}^{m}P_j\alpha_1\geq\sum_{j=1}^{m}P_j\alpha_j\geq1-D.$$ Thus, $$d_1=\mathop{\min}_{1-D\leq\alpha_{1}\leq1}\mathop{\min}_{\alpha_2,\alpha_3,\cdots,\alpha_m}\log\frac{\sum_{j=2}^{m}(\alpha_j-\delta)}{1-\alpha_1},$$ where $d_1:=0$ for $\alpha_1=1.$ For fixed $\alpha_{1}\in[1-D,1]$, consider the following optimization problem,
\begin{eqnarray}
\nonumber\lefteqn{d_2(\alpha_1):=\min_{\alpha_2,\alpha_3,\cdots,\alpha_m}\sum_{j=2}^{m}(\alpha_j-\delta)}\\
\nonumber&&\text{subject to}\\
\label{op:adp2}&&(1)\text{ }\sum_{j=2}^{m}P_j\alpha_j\geq1-D-P_{1}\alpha_{1},\\
\nonumber&&(2)\text{ }\sum_{j=2}^{m}(\alpha_j-\delta)\geq1-\alpha_{1},\\
\nonumber&&(3)\text{ }\alpha_{j}\leq\alpha_{j-1}\text{ for }2\leq j\leq m+1.
\end{eqnarray}
The optimal values of optimization problem \eqref{op:adp1} and \eqref{op:adp2} satisfy $d_1=\min_{1-D\leq\alpha_{1}\leq1}\log\frac{d_2(\alpha_1)}{1-\alpha_1}$. 
Denote $\beta_{j}:=\alpha_{j}-\alpha_{j+1}$ for $1\leq j\leq m-1$. Then problem \eqref{op:adp2} can be rewritten as the following optimization problem.
\begin{eqnarray}
\nonumber\lefteqn{d_3(\alpha_1):=\max_{\beta_1,\beta_2,\cdots,\beta_{m-1}}\sum_{j=1}^{m-1}(m-j)\beta_j}\\
\nonumber&&\text{subject to}\\
\label{op:adp3}&&(1)\text{ }\sum_{j=1}^{m-1}D^{(m-j)}\beta_j\leq\alpha_{1}-1+D,\\
\nonumber&&(2)\text{ }\sum_{j=1}^{m-1}(m-j)\beta_j\leq m(\alpha_{1}-\delta)-1+\delta,\\
\nonumber&&(3)\text{ }\sum_{j=1}^{m-1}\beta_{j}\leq\alpha_1-\delta,\\
\nonumber&&(4)\text{ }\beta_{j}\geq0\text{ for }1\leq j\leq m-1.
\end{eqnarray}
Moreover, $d_2(\alpha_1)=(m-1)(\alpha_{1}-\delta)-d_{3}(\alpha_1).$ The dual of linear program \eqref{op:adp3} is the minimization linear program as below.
\begin{eqnarray}
\nonumber\lefteqn{\tilde{d}_3(\alpha_1):=\min_{\gamma_1,\gamma_2,\gamma_{3}}(\alpha_{1}-1+D)\gamma_1+(m(\alpha_{1}-\delta)-1+\delta)\gamma_2}\\\nonumber&&\ \ \ \ \ \ \ \ \ \ \ \ \ \  +(\alpha_{1}-\delta)\gamma_3\\
\label{op:adp33}&&\text{subject to}\\
\nonumber&&(1)\text{ }D^{(j)}\gamma_1+j\gamma_2+\gamma_3\geq j\text{ for }1\leq j\leq m-1\\
\nonumber&&(2)\text{ }\gamma_1,\gamma_2,\gamma_{3}\geq0.
\end{eqnarray}
By strong duality theorem, we have ${d}_3(\alpha_1)=\tilde{d}_3(\alpha_1)$. Fix $\gamma_2\geq0$, we define the following linear program.
\begin{eqnarray}
\nonumber\lefteqn{d_4(\alpha_1,\gamma_2):=\min_{\gamma_1,\gamma_{3}}(\alpha_{1}-1+D)\gamma_1+(\alpha_1-\delta)\gamma_3}\\
\label{op:adp4}&&\text{subject to}\\
\nonumber&&(1)\text{ }D^{(j)}\gamma_1+\gamma_3\geq j(1-\gamma_2)\text{ for }1\leq j\leq m-1,\\
\nonumber&&(2)\text{ }\gamma_1,\gamma_{3}\geq0.
\end{eqnarray}
Hence $$\tilde{d}_3(\alpha_1)=\min_{\gamma_2\geq0}\left\{d_4(\alpha_1,\gamma_2)+(m(\alpha_1-\delta)-1+\delta)\gamma_2\right\}.$$

Case I. $\gamma_2\geq1$. Let us consider the optimization problem \eqref{op:adp4}, where $d_4(\alpha_1,\gamma_2)=\min_{\gamma_1,\gamma_3}(\alpha_{1}-1+D)\gamma_1+(\alpha_1-\delta)\gamma_3$ with $\gamma_1$, $\gamma_{3}\geq0$. Thus, $d_4(\alpha_1,\gamma_2)=0.$

Case II. $0\leq\gamma_2<1$. Let $l_j$ denote the line $\gamma_3=-D^{(j)}\gamma_1+j(1-\gamma_2)$ for $1\leq j\leq m-1$. Assume that $l_j$ intersects with $l_i$ at $Z_{ij}(u_{ij},v_{ij})$ for $i\neq j$. It is easy to calculate the coordinate of the intersection point, i.e.,  
$$
\begin{cases}
u_{ij}=\frac{(1-\gamma_2)(i-j)}{D^{(i)}-D^{(j)}},\\
v_{ij}=\frac{(1-\gamma_2)(jD^{(i)}-iD^{(j)})}{D^{(i)}-D^{(j)}}.
\end{cases}
$$ 
We claim that for any $j\in [1,m-1]$,
\begin{equation}\label{eq:claim}
\begin{cases}
u_{1j}>u_{2j}>\cdots>u_{j-1,j}>u_{j+1,j}>\cdots>u_{m-1,j}>0,\\
0<v_{1j}<v_{2j}<\cdots<v_{j-1,j}<v_{j+1,j}<\cdots<v_{m-1,j}.
\end{cases}
\end{equation}
Therefore, the feasible region of \eqref{op:adp4} is unbounded, and corner points are located from top left to bottom right as follows.
$$
(0,(m-1)(1-\gamma_2))\to Z_{m-1,m}\to \cdots Z_{2,3}\to  Z_{1,2}\to (\frac{1-\gamma_2}{P_{m}},0).
$$

We now prove (\ref{eq:claim}). By the definition of $D^{(k)}$, $D^{(k)}$ is increasing in $k$, and hence $u_{ij}$ is positive. If $j>i$ then $jD^{(i)}-iD^{(j)}=(j-i)D^{(i)}-i(D^{(j)}-D^{(i)})=\sum_{a=m-i+1}^{m}\sum_{b=m-j+1}^{m-i}(P_{a}-P_{b})<0$. Thus $v_{ij}$ is positive.

For $k\neq j$ and $1\leq i<k\leq m$, 
\beqs
\lefteqn{\frac{u_{kj}-u_{ij}}{1-\gamma_2}=\frac{k-j}{D^{(k)}-D^{(j)}}-\frac{i-j}{D^{(i)}-D^{(j)}}}\\
&&=\frac{(i-k)D^{(j)}+(k-i)D^{(i)}+(i-j)(D^{(i)}-D^{(k)})}{(D^{(k)}-D^{(j)})(D^{(i)}-D^{(j)})}\\
&&=\frac{(k-i)(D^{(i)}-D^{(j)})-(i-j)(D^{(k)}-D^{(i)})}{(D^{(k)}-D^{(j)})(D^{(i)}-D^{(j)})}\\
&&=
\begin{cases}
\frac{\sum_{a=m-i+1}^{m-j}\sum_{b=m-k+1}^{m-i}(P_{a}-P_{b})}{(D^{(k)}-D^{(j)})(D^{(i)}-D^{(j)})}&\text{ if }k>i>j,\\
\frac{\sum_{a=m-j+1}^{m-k}\sum_{b=m-k+1}^{m-i}(P_{b}-P_{a})}{(D^{(k)}-D^{(j)})(D^{(i)}-D^{(j)})}&\text{ if }j>k>i,\\
\frac{\sum_{a=m-j+1}^{m-i}\sum_{b=m-k+1}^{m-j}(P_{b}-P_{a})}{(D^{(k)}-D^{(j)})(D^{(i)}-D^{(j)})}&\text{ if }k>j>i.
\end{cases}
\\
&&<0.
\eeqs
Thus, $u_{ij}$ is decreasing in $i$. On the other hand, line $l_{j}$ with negative slope $-D^{(j)}$ passes point $Z_{ij}$ for any $1\leq i\leq m$, which implies $v_{ij}$ is increasing in $i$.

The slope of the objective function in \eqref{op:adp4} determines which corner point will be reached last. Hence, we divide $[1-D,\infty]$ into $m$ closed intervals. For $\frac{1-D-\delta D^{(k-1)}}{1-D^{(k-1)}}\leq\alpha_{1}\leq\frac{1-D-\delta D^{(k)}}{1-D^{(k)}}$, $k=1$, $2$, $\cdots$, $m$, we have $-D^{(k)}\leq-(\alpha_{1}-1+D)/{(\alpha_{1}-\delta)}\leq-D^{(k-1)}$. Thus, $Z_{k-1,k}$ is optimal solution of \eqref{op:adp4} and 
\beqs
\lefteqn{d_4(\alpha_1,\gamma_2)=}\\
&&\frac{1-\gamma_2}{P_{m-k+1}}\left((\alpha_1-1+D)+(\alpha_{1}-\delta)((k-1)P_{m-k+1}-D^{(k-1)})\right),
\eeqs
for $0\leq\gamma_2<1$. 
Let 
\beqs
\lefteqn{c_{k}(\alpha_1):=}\\
&&\frac{1}{P_{m-k+1}}\left((\alpha_1-1+D)+(\alpha_{1}-\delta)((k-1)P_{m-k+1}-D^{(k-1)})\right)
\eeqs
Then for $\frac{1-D-\delta D^{(k-1)}}{1-D^{(k-1)}}\leq\alpha_{1}\leq\frac{1-D-\delta D^{(k)}}{1-D^{(k)}}$ with $k$ satisfying $D>(1-\delta)D^{(k-1)}$,  we have
\beqs
\lefteqn{\tilde{d}_3(\alpha_1)=\min_{\gamma_2\geq0}\left\{d_4(\alpha_1,\gamma_2)+(m(\alpha_{1}-\delta)-1+\delta)\gamma_2\right\}}\\
&&=\min_{0\leq\gamma_2\leq1}\left\{c_{k}(\alpha_1)(1-\gamma_2)+(m(\alpha_{1}-\delta)-1+\delta)\gamma_2\right\}\\
&&=\min\left\{c_{k}(\alpha_1),m(\alpha_{1}-\delta)-1+\delta\right\}.\\
\eeqs
Thus, 
\beqs
& &{d}_3(\alpha_1)=\tilde{d}_3(\alpha_1)=\min\left\{c_{k}(\alpha_1),m(\alpha_{1}-\delta)-1+\delta\right\}\\
&\Rightarrow& d_2(\alpha_1)=\max\left\{1-\alpha_1,(m-1)(\alpha_{1}-\delta)-c_{k}(\alpha_1)\right\}\\
&\Rightarrow& d_1=\min_{\alpha_1,k}\log\max\left\{1,\frac{(m-1)(\alpha_{1}-\delta)-c_{k}(\alpha_1)}{1-\alpha_1}\right\}\\
&\Rightarrow& \epsilon_{\delta}^{*}(P,D)=\min_{D>(1-\delta)D^{(k-1)}}\log\frac{(m-k)(1-D-\delta)}{D-(1-\delta)D^{(k-1)}}.
\eeqs
In a summary, we have that
if $0 <D < (1-\delta)D^{(1)}$, then $$\epsilon^{*}_{\delta}(P,D)=\log\frac{(m-1)(1-D-\delta)}{D};$$
if $(1-\delta)D^{(k-1)} \leq D < (1-\delta)D^{(k)}$ for $2\leq k\leq m-1$, then
$$\epsilon^{*}_{\delta}(P,D)=\min_{1\leq j\leq k-1}\log\frac{(m-j)(1-D-\delta)}{D-(1-\delta)D^{(j-1)}};$$
if $(1-\delta)D^{(m-1)} \leq D \leq 1$, then
$$\epsilon^{*}_{\delta}(P,D)=0.$$
\end{proof}

Next, let us consider the case when the input distribution belongs to some set $\mathcal{P}$ of Class II. 

\begin{prop}
The approximate differential privacy-distortion with $\mathcal{P}$ of Class II is 
$$\epsilon^{*}_{\delta}(\mathcal{P},D)=\min_{1-D\leq\alpha_1\leq1}\log\frac{(m-1)(\alpha_1-\delta)-\min_{\gamma_2\geq0}d(\alpha_1,\gamma_2)}{1-\alpha_1}$$
where $d(\alpha_1,\gamma_2)$ is the optimal value of PV.
\end{prop}
\begin{proof}
Recall the optimization problem \eqref{op:adp4}, we can also obtain the approximate differential-privacy function as the solution to the following optimization problem.
\beqs
&&\min_{\gamma_{\mathcal{P}},\gamma_{3}}(\alpha_{1}-1+D)\sum_{P\in\mathcal{P}}\gamma_P+(\alpha_1-\delta)\gamma_3\\
&&\text{subject to}\\
&&(1)\text{ }\sum_{P\in\mathcal{P}}D_{P}^{(j)}\gamma_{P}+\gamma_3\geq j(1-\gamma_2)\text{, }1\leq j\leq m-1\\
&&(2)\text{ }\gamma_P,\gamma_{3}\geq0\text{, }\forall P\in\mathcal{P}
\eeqs
where $D^{(k)}_{P}=\sum_{j=m-k+1}^{m}P_{j}$ for the corresponding $P$ in $\mathcal{P}$. Thus, we have the result.

\end{proof}

By taking $\delta=0$, we get the following result of 
 differential privacy-distortion function for given input distribution, which 
 is stronger than that in  \cite{Kalantari2018}, where it only provides the optimization problem instead of the solution.  
\begin{thm}\label{thm:DP1}
Given distribution $P$ with $P_1\geq P_2 \geq \cdots \geq P_{m}$, we have
$$\epsilon^{*}_{DP}(P,D)=\max\left\{0,\min_{D>D^{(k-1)}}\log\frac{(m-k)(1-D)}{D-D^{(k-1)}}\right\}.$$ 
\end{thm}

A recent paper \cite{Saeidian2021} studied the privacy-utility trade-off by taking maximal leakage as the privacy measure. We regain the results by coloring argument which we put in Appendix \ref{app:ml}. 








\begin{thm}\label{thm:DP4}
Let $D^{(k)}=\sum_{j=m-k+1}^{m}P_{j}$ for $0\leq k\leq m$. Then 
$$\epsilon^{*}_{ML}(P,D)=\max\left\{0,\log\left(m-k-\frac{D-D^{(k)}}{D^{(k)}-D^{(k-1)}}\right)\right\}$$
if $D^{(k-1)}<D\leq D^{(k)}$ and $1\leq k\leq m$.
\end{thm}
An equal expression of Theorem \ref{thm:DP4} showing smallest distortion with fixed leakage was given by Saeidian et al.\cite{Saeidian2021}. Assume that the prior belongs to some set $\mathcal{P}$ of Class II. Recall the definition of $D^{*}(\mathcal{P},\epsilon)$ in Definition \ref{def:pri-dis}, we have a straightforward result as follows.

\begin{cor}[\cite{Saeidian2021}]
Given $\epsilon\geq0$, 
$$D_{ML}^{*}(\mathcal{P},\epsilon)=\max_{P\in\mathcal{P}}D_{ML}^{*}(P,\epsilon)$$ where $D_{ML}^{*}(P,\epsilon)=D^{(m-k)}+P_{k+1}(k-e^{\epsilon})$ if $\log k\leq \epsilon<\log (k+1)$ and $1\leq k\leq m.$ 
\end{cor}

\subsection{Illustration of Local Results}
In this subsection, we demonstrate the results in local privacy framework by figures. Our results are first applied to binary datasets. In particular, each attribute associated with individuals is the answer of some binary classification problem, that is a yes/no question or a setting with $0$-$1$ outcome. In general, $\mathcal{X}=\{0,1\}$ and $m=2$. For local privacy mechanism, we compare the privacy-distortion functions with input distribution unknown by taking $m=2$ into Theorem \ref{thm:t1,c1}.
\begin{figure}
    \centering
    \includegraphics[width=0.5\textwidth]{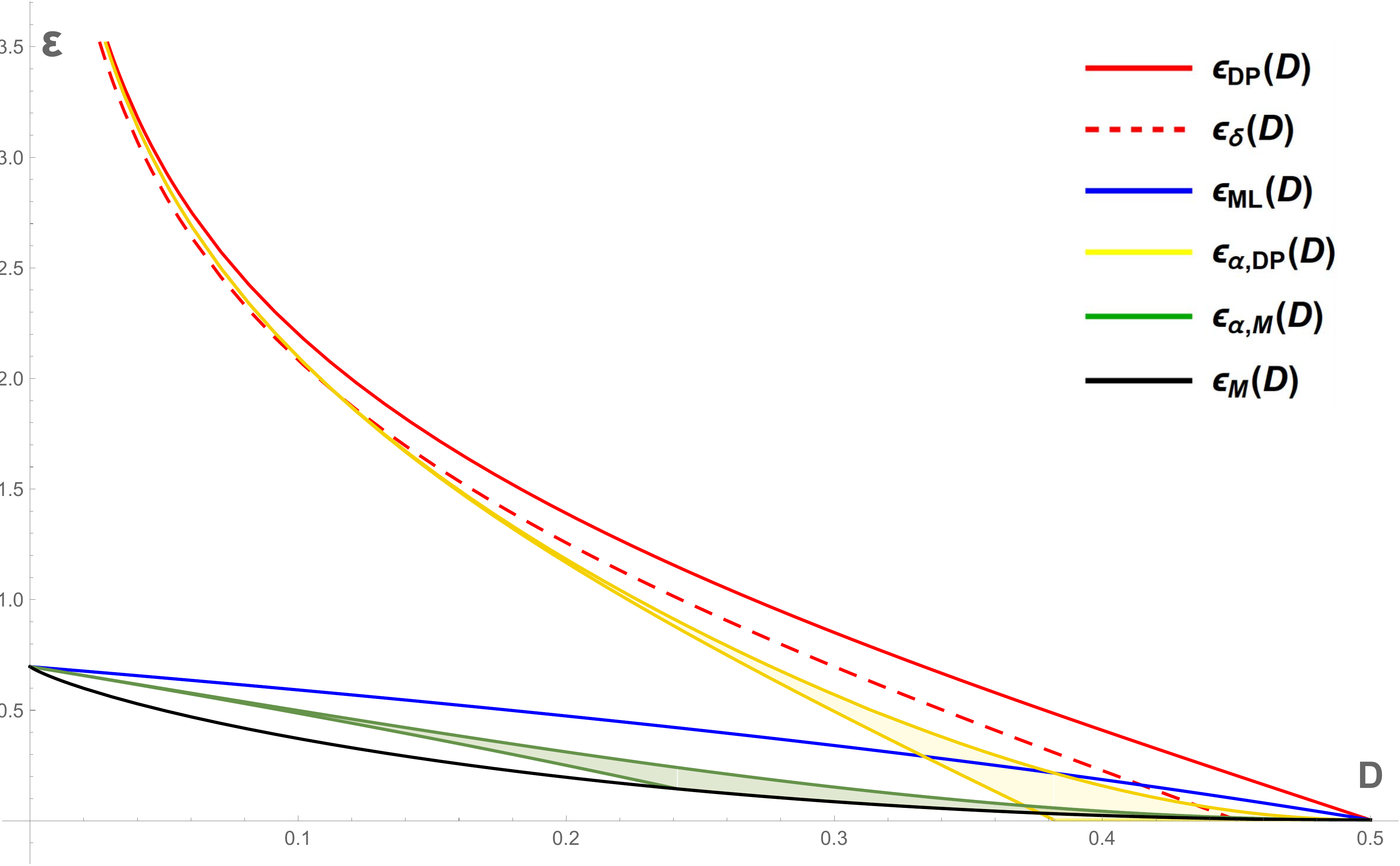}
    \caption{Privacy-Distortion trade-off curves for differential privacy ($\epsilon_{DP}(D)$), $0.1$-approximate differential privacy ($\epsilon_{0.1}(D)$), maximal leakage ($\epsilon_{ML}(D)$),  and mutual information ($\epsilon_{M}(D)$); Privacy-Distortion trade-off ranges for R\'{e}nyi differential privacy ($\epsilon_{2,DP}(D)$) and Sibson mutual information ($\epsilon_{2,M}(D)$) with $\alpha=2$.}
    \label{fig:1c1}
\end{figure}
It is shown in Fig. \ref{fig:1c1} that high accuracy requirement (i.e., $D<0.5$) leads to the phenomena-the privacy costs for divergence-based measures (i.e., $\epsilon_{DP}$, $\epsilon_{\delta}$ and $\epsilon_{\alpha,DP}$) are numerically higher than the ones for mutual information (i.e., $\epsilon_{ML}$, $\epsilon_{\alpha,M}$ and $\epsilon_{M}$). 
This is reasonable since divergence-based privacy pays more attention to individuals than the collectives. Notably, the  randomized response mechanism $Q_D$ is always an optimal choice.

We now turn to the privacy-distortion functions with the knowledge of prior distribution. To demonstrate the power of Theorem \ref{thm:DP2}, we set a certain example with $m=4$ and prior distribution $P=(P_1,P_2,P_3,P_4)=(0.4,0.3,0.2,0.1)$.  
\begin{figure}
    \centering
    \includegraphics[width=0.5\textwidth]{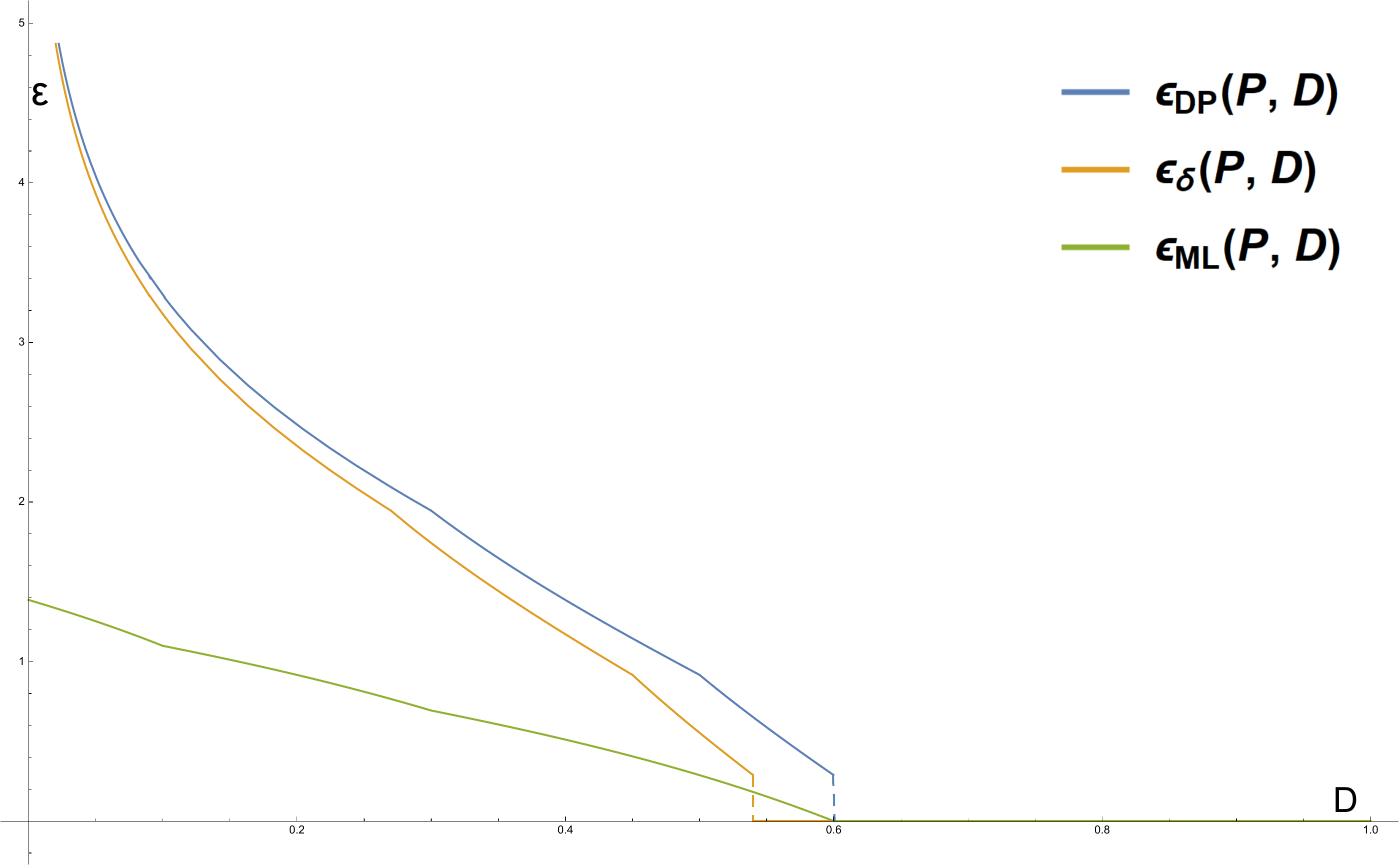}
    \caption{Privacy-Distortion trade-off curves for differential privacy ($\epsilon_{DP}(P,D)$), $0.1$-approximate differential privacy ($\epsilon_{0.1}(P,D)$), and maximal leakage ($\epsilon_{ML}(P,D)$).}
    \label{fig:1c2}
\end{figure}
As shown in Fig. \ref{fig:1c2}, maximal leakage is the most relaxed privacy notion among these three popular measures. Notably, the curves for differential privacy and approximate differential privacy jump to zero when $D=0.6$ and $0.54$, respectively. This is because the special mechanism $Q_{\delta}$ mentioned in the proof of Lemma \ref{lem:02} is $(P,D)$-valid if and only if $D\geq (1-\delta)(1-P_1).$


\subsection{Parallel Composition}\label{sec:5}

To create a database consisting of data with local privacy guarantees, the data server usually collects the i.i.d. samples from all possible individual uploads (see Fig. \ref{fig:composition}). In other words, the input space $\mathcal{X}^n$ contains the elements which are independently and identically distributed samples from $\mathcal{X}=\{1,2,\cdots,m\}$ with full-support distribution $P_X$, i.e., $$P_X(\boldsymbol{x})=P_X(x_1,\cdots,x_n)=\prod_{j=1}^{n}P_{x_j}.$$ 
Then the global privacy mechanisms on $\mathcal{X}^n$ by $Q$ on $\mathcal{X}$ is given as follows. 
$$Q_{Y|X}(\boldsymbol{y}|\boldsymbol{x})=Q_{Y|X}(y_1,\cdots,y_n|x_1,\cdots,x_n)=\prod_{j=1}^{n}Q_{Y|X}(y_j|x_j).$$
Due to the independence of sampling, we have $$\mathbb{E}_{P,Q}[d(\boldsymbol{x},\boldsymbol{y})]=n\mathbb{E}_{P,Q}[d(x,y)].$$

\begin{figure}
    \centering
    \includegraphics[width=0.5\textwidth]{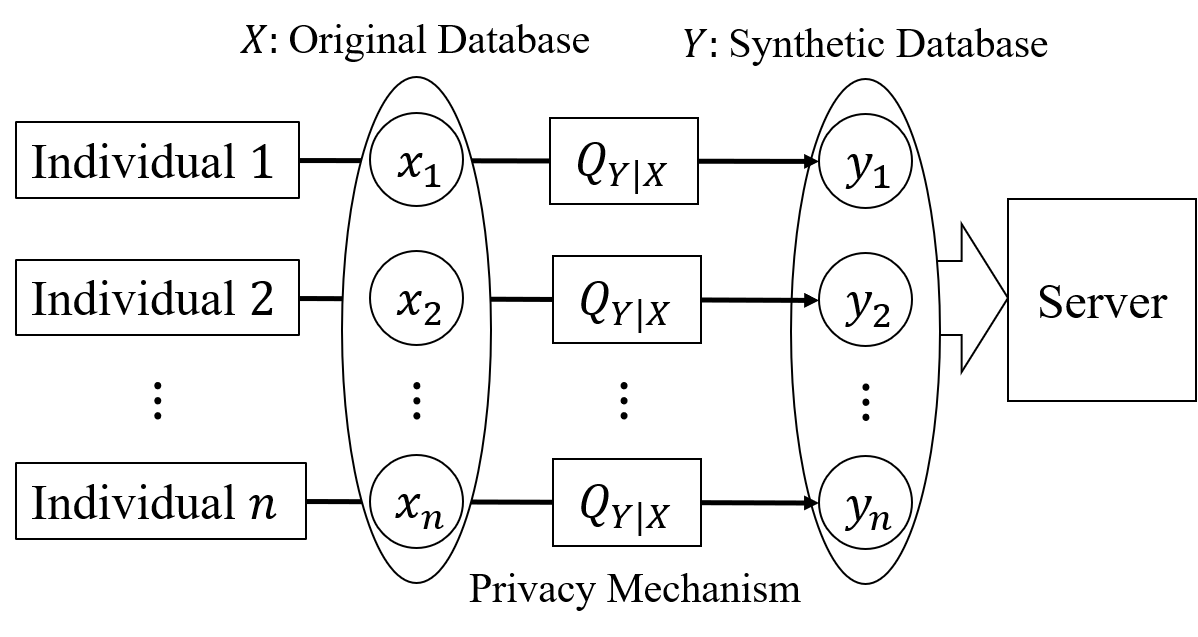}
    \caption{Parallel Composition: the data server collects the i.i.d. samples from all possible individual uploads with local privacy guarantees.}
    \label{fig:composition}
\end{figure}

\begin{lem}\label{lem:basiccomp}
Let $\hat\epsilon^{*}(\mathcal{P},D)$ be the privacy-distortion function over $\mathcal{X}^n$. Then

(1) $\hat\epsilon^{*}(\mathcal{P},D)=\epsilon^{*}(\mathcal{P},D/n)$ for differential privacy and R\'{e}nyi differential privacy.

(2) $\hat\epsilon^{*}(\mathcal{P},D)=n\epsilon^{*}(\mathcal{P},D/n)$ for maximal information, maximal leakage, Sibson mutual information and mutual information.
\end{lem}

\begin{proof}[Proof of Lemma \ref{lem:basiccomp}]
By the relations among these privacy notions (See Lemma \ref{remark:implication} in Appendix \ref{app:bpp}), it suffices to prove the cases for maximal information and the $\alpha$-type notions. Let $Q$ be a mechanism from $\mathcal{X}^n$ to $\mathcal{Y}^n$.

1)(Maximal information)
\beqs
\lefteqn{\hat\epsilon_{MI}(Q)=\max_{\boldsymbol{x},\boldsymbol{y}}\log\frac{Q(\boldsymbol{y}|\boldsymbol{x})}{\sum_{\boldsymbol{x}}P(\boldsymbol{x})Q(\boldsymbol{y}|\boldsymbol{x})}}\\
&&=\max_{\substack{x_1,\cdots,x_n\\y_1,\cdots,y_n}}\log\frac{\prod_{j=1}^{n}Q(y_j|x_j)}{\sum_{x_1,\cdots,x_n}\prod_{j=1}^{n}P(x_j)Q(y_j|x_j)}\\
&&=\max_{\substack{x_1,\cdots,x_n\\y_1,\cdots,y_n}}\log\frac{\prod_{j=1}^{n}Q(y_j|x_j)}{\prod_{j=1}^{n}\sum_{x_j}P(x_j)Q(y_j|x_j)}\\
&&=\sum_{j=1}^{n}\max_{x_j,y_j}\log\frac{Q(y_j|x_j)}{\sum_{x_j}P(x_j)Q(y_j|x_j)}\\
&&=n\epsilon_{MI}(Q).
\eeqs

2)(R\'{e}nyi differential privacy)
\beqs
&&\exp((\alpha-1)\hat\epsilon_{\alpha,DP}(Q))\\
&&=\max_{d(\boldsymbol{x},\boldsymbol{x}^{\prime})=1}\sum_{\boldsymbol{y}}Q^{\alpha}(\boldsymbol{y}|\boldsymbol{x})Q^{1-\alpha}(\boldsymbol{y}|\boldsymbol{x}^{\prime})\\
&&=\max_{\substack{x_1,\cdots,x_{n-1}\\x\neq x^{\prime}}}\sum_{y_1,\cdots,y_n}\prod_{j=1}^{n}Q^{\alpha}(y_{n}|x)Q^{1-\alpha}(y_n|x^{\prime})\\
&&=\max_{\substack{x_1,\cdots,x_{n-1}\\x\neq x^{\prime}}}\prod_{j=1}^{n-1}\sum_{y_j}Q(y_j|x_j)\sum_{y}Q^{\alpha}(y|x)Q^{1-\alpha}(y|x^{\prime})\\
&&=\exp((\alpha-1)\epsilon_{\alpha,DP}(Q)).
\eeqs

3)(Sibson mutual information)
\beqs
\lefteqn{\hat\epsilon_{\alpha,M}(Q)=\frac{\alpha}{\alpha-1}\log\sum_{\boldsymbol{y}}\left(\sum_{\boldsymbol{x}}P(\boldsymbol{x})Q^{\alpha}(\boldsymbol{y}|\boldsymbol{x})\right)^{1/\alpha}}\\
&&=\frac{\alpha}{\alpha-1}\log\sum_{y_1,\cdots,y_n}\left(\sum_{x_1,\cdots,x_n}\prod_{j=1}^{n}P(x_{j})Q^{\alpha}(y_j|x_j)\right)^{1/\alpha}\\
&&=\frac{\alpha}{\alpha-1}\log\prod_{j=1}^{n}\sum_{y_j}\left( \sum_{x_j}P(x_j)Q^{\alpha}(y_j|x_j)\right)^{1/\alpha}\\
&&=n\epsilon_{\alpha,M}(Q).
\eeqs

Since $\mathbb{E}_{P,Q}[d(\boldsymbol{x},\boldsymbol{y})]=n\mathbb{E}_{P,Q}[d(x,y)]$, the results above complete the proof.
\end{proof}

Based on
the additivity of composition and Theorem \ref{thm:t1,c1} , we get the following results  directly.

\begin{thm}\label{thm:bn,c1}
For any $\mathcal{P}$ of Class I, we have

(1) (Differential privacy)
$$\hat\epsilon^{*}_{DP}(\mathcal{P},D)=
\begin{cases}
\log\frac{(m-1)(n-D)}{D}, & 0<D <\frac{n(m-1)}{m},\\
0, & \frac{n(m-1)}{m} \leq D \leq 1;
\end{cases}$$

(2) (Approximate differential privacy)
\beqs
&&\max\left\{0,\log(m-1)\frac{(n(1-\delta m^{n-1})-D)}{D}\right\}\\&\leq&\hat\epsilon^{*}_{\delta}(\mathcal{P},D)\\
&\leq&\max\left\{0,\log(m-1)\frac{n-D}{D}(1-\delta(1-\frac{D}{n})^{-n})\right\};
\eeqs

(3) (Maximal information)
$$
\begin{cases}
n\log m(1-\frac{D}{n})\leq\hat\epsilon^{*}_{MI}(\mathcal{P},D)\leq \hat\epsilon_{MI}(Q_D), & 0<D <\frac{n(m-1)}{m},\\
\hat\epsilon^{*}_{MI}(\mathcal{P},D)=0, & \frac{n(m-1)}{m} \leq D \leq 1
\end{cases}$$
where $\hat\epsilon_{MI}(Q_D)=n\log\frac{n-D}{P^{*}_{\mathcal{P}}(n-D-\frac{D}{m-1})+\frac{D}{m-1}}$ and $P^{*}_{\mathcal{P}}=\min_{1\leq i\leq m}$ $\{P_{i}|P\in\mathcal{P}\};$

(4) (Maximal leakage)
$$\hat\epsilon^{*}_{ML}(\mathcal{P},D)=
\begin{cases}
n \log m(1-\frac{D}{n}), & 0<D <\frac{n(m-1)}{m},\\
0, & \frac{n(m-1)}{m} \leq D \leq 1;
\end{cases}$$

(5) (R\'{e}nyi differential privacy)
\beqs
&&\max\left\{0,\frac{n(m-1)(1-\frac{D}{n})^{\alpha/(\alpha-1)}}{D}\right\}\\&\leq&\hat\epsilon^{*}_{\alpha,DP}(\mathcal{P},D)\\&\leq&\max\left\{0,\hat\epsilon_{\alpha,DP}(Q_D)\right\}
\eeqs
where \beqs
\hat\epsilon_{\alpha,DP}(Q_D)=\frac{1}{\alpha-1}\log\frac{D}{n(m-1)}\Big(m-2+\left(\frac{(n-D)(m-1)}{D}\right)^{\alpha}&&\\
+\left(\frac{(n-D)(m-1)}{D}\right)^{1-\alpha}\Big);&&
\eeqs

(6) (Sibson mutual information)
\beqs
&&\max\left\{\hat\epsilon^{*}_{M}(\mathcal{P},D),n\log m(1-\frac{D}{n})^{\frac{\alpha}{\alpha-1}}\right\}\\&\leq&\hat\epsilon^{*}_{\alpha,M}(\mathcal{P},D)\\&\leq&\max\left\{0,\hat\epsilon_{\alpha,M}(Q_D)\right\}
\eeqs
where $$\hat\epsilon_{\alpha,M}(Q_D)=n\log m \left((1-\frac{D}{n})^{\alpha}+(m-1)\left(\frac{D}{n(m-1)}\right)^{\alpha}\right)^{\frac{1}{\alpha-1}};$$

(7) (Mutual information)
$$\tilde{\epsilon}^{*}_{M}(\mathcal{P},D)=\max\left\{0,n\log m(1-\frac{D}{n})+D\log\left(\frac{D}{(m-1)(n-D)}\right)\right\}.$$

\end{thm}

\begin{proof}[Proof of (2) in Theorem \ref{thm:bn,c1}]
For given $\delta\in(0,1)$,
\beqs
\lefteqn{\hat{\epsilon}_{\delta}(Q)=\max_{\boldsymbol{y},d(\boldsymbol{x},\boldsymbol{x}^{\prime})=1}\log\frac{Q(\boldsymbol{y}|\boldsymbol{x})}{Q(\boldsymbol{y}|\boldsymbol{x}^{\prime})}}\\
&&=\max_{y,x\neq x^{\prime}}\log\frac{Q(y|x)-\delta Q_{*}^{1-n}}{Q(y|x^{\prime})}=\epsilon_{\delta Q_{*}^{1-n}}(Q)
\eeqs
where $Q_{*}=\max_{i,j}Q(j|i).$ Thus, $$\epsilon^{*}_{\delta m^{n-1}}(\mathcal{P},\frac{D}{n})\leq \hat{\epsilon}^{*}_{\delta }(\mathcal{P},D)\leq \epsilon^{*}_{\delta}(\mathcal{P},\frac{D}{n}).$$
\end{proof}


\begin{thm}let $P$ be a given prior distribution such that $P_1\geq P_2\geq \cdots \geq P_m$. 
Let $D^{(k)}=\sum_{j=m-k+1}^{m}P_{j}$ for $0\leq k\leq m.$ Then

(1) (Approximate differential privacy) $$\epsilon_{\delta}^{*}(P,nD)=\max\left\{0,\min_{D>(1-\delta)D^{(k-1)}}\log\frac{(m-k)(1-D-\delta)}{D-(1-\delta)D^{(k-1)}}\right\}.$$

(2) (Maximal leakage)
$$\epsilon^{*}_{ML}(P,nD)=\max\left\{0,n\log\left(m-k-\frac{D-D^{(k)}}{D^{(k)}-D^{(k-1)}}\right)\right\}$$
if $D^{(k-1)}<D\leq D^{(k)}$ and $1\leq k\leq m$.
\end{thm}

\section{Conclusion}\label{sec:6}

This paper investigates the privacy-utility trade-off problem for seven popular privacy measures in two different scenarios: global and local privacy. For both global and local privacy, our results 
provide some upper and lower bounds on privacy-distortion function, which 
reveals the relationships between privacy and distortion
in a quantitative  way. 
In particular, with known prior distributions, we obtain the analytical closed form of privacy-distortion function for approximate differential privacy. To the best of our knowledge, this is the first result on the privacy-uitility tradeoff for
approximate differential privacy.



In addition to the privacy measures we used here, it is still possible to
consider other privacy measures, such as the privacy measure defined by Arimoto $\alpha$-mutual information and
$f$-mutual information. 
Moreover, the application privacy-utility trade-offs in 
machine learning is an important problem, which we leave it for further study.


 \section{Acknowledgments}\
K. Bu thanks the support from ARO Grant W911NF-19-1-0302 and the ARO MURI Grant W911NF-20-1-0082.

\clearpage

\begin{appendix}

\subsection{Basic Properties of Privacy}\label{app:bpp}
To clearly show the properties of different notions for any given mechanism $Q$, we turn to the following definition equivalent to Definition \ref{def:prav}.

\begin{Def}
Let $P_X$ be a fixed distribution with full support on $\mathcal{X}$. 

(1) (Maximal divergence) The differential privacy loss of $Q$ is defined as $$\epsilon_{DP}(Q):=\max_{\substack{y\in\mathcal{Y}, d(x,x^{\prime})=1\\Q_{Y|X}(y|x^{\prime})>0}}\log\frac{Q_{Y|X}(y|x)}{Q_{Y|X}(y|x^{\prime})}.$$

(2) (Approximate maximal divergence) Given a postive real number $\delta<1$, the approximate differential privacy loss of $Q$ is defined as $$\epsilon_{\delta}(Q):=\max\left\{0,{\max_{\substack{y\in\mathcal{Y}, d(x,x^{\prime})=1\\Q_{Y|X}(y|x^{\prime})>0}}\log\frac{Q_{Y|X}(y|x)-\delta}{Q_{Y|X}(y|x^{\prime})}}\right\}.$$

(3) (Maximal information) The maximal information between $X$ and $Q(X)$ is denoted by $$\epsilon_{MI}(Q):=I_{\infty}(X;Q(X))=\max_{x\in\mathcal{X},y\in\mathcal{Y}}\log\frac{Q_{Y|X}(y|x)}{\sum_{x\in\mathcal{X}}P_{X}(x)Q_{Y|X}(y|x)}.$$

(4) (Maximal leakage) The maximal leakage from $X$ to $Q(X)$ is denoted by $$\epsilon_{ML}(Q):=\mathcal{L}(X\to Q(X))=\log\sum_{y\in\mathcal{Y}}\max_{x\in\mathcal{X}}Q_{Y|X}(y|x).$$

(5) (R\'{e}nyi divergence) Given a postive real number $\alpha>1$, the R\'{e}nyi differential privacy loss of $Q$ is defined as 
\beqs
\epsilon_{\alpha,DP}(Q)&:=&\max_{d(x,x^{\prime})=1}D_{\alpha}\left({Q_{Y\vert X}(\cdot\vert x)\Vert Q_{Y\vert X}(\cdot\vert x^{\prime})}\right)\\
&=&\max_{d(x,x^{\prime})=1}\frac{1}{\alpha-1}\log\sum_{y\in\mathcal{Y}}Q_{Y|X}^{\alpha}(y|x)Q_{Y|X}^{1-\alpha}(y|x^{\prime}).
\eeqs

(6) (Sibson's mutual information) Given a postive real number $\alpha>1$, the $\alpha$-mutual informationa of $Q$ is denoted as 
\beqs
\epsilon_{\alpha,M}(Q)&:=&I_{\alpha}\left({X;Q(X)}\right)\\
&=&\frac{\alpha}{\alpha-1}\log\sum_{y\in\mathcal{Y}}\left({\sum_{x\in\mathcal{X}}P_X(x)Q_{Y|X}^{\alpha}(y|x)}\right)^{1/\alpha}.
\eeqs

(7) (Mutual information) The mutual information of $Q$ is denoted as
\beqs
\epsilon_{M}(Q)&:=&I\left({X;Q(X)}\right)\\
&=&\sum_{x,y}P_X(x)Q_{Y|X}(y|x)\log\frac{Q_{Y|X}(y|x)}{\sum_{x}P_{X}(x)Q_{Y|X}(y|x)}.
\eeqs

\end{Def}

Then we have the following well-known results.

\begin{lem}\label{remark:implication}
Relations and implications.

(1) The privacy losses $\epsilon_{DP}(Q)$, $\epsilon_{\delta}(Q)$, $\epsilon_{ML}(Q)$ and $\epsilon_{\alpha,DP}(Q)$ do not depend on $P$. 

(2) The privacy losses $\epsilon_{DP}(Q)$, $\epsilon_{\delta}(Q)$ and $\epsilon_{\alpha,DP}(Q)$ provide worst-case privacy guarantees.

(3) The privacy losses $\epsilon_{MI}(Q)$, $\epsilon_{ML}(Q)$, $\epsilon_{\alpha,M}(Q)$ and $\epsilon_{M}(Q)$ provide global privacy guarantees.

(4) For $0\leq\delta_1\leq\delta_2$, $\epsilon_{\delta_2}(Q)\leq\epsilon_{\delta_1}(Q).$ In particular, $\epsilon_{\delta=0}(Q)=\epsilon_{DP}(Q).$

(5) $\lim_{\alpha\to \infty}\epsilon_{\alpha,DP}(Q)=\epsilon_{DP}(Q).$

(6) $\lim_{\alpha\to 1}\epsilon_{\alpha,M}(Q)=\epsilon_{M}(Q).$

(7) $\lim_{\alpha\to \infty}\epsilon_{\alpha,M}(Q)=\epsilon_{ML}(Q).$

(8) $\epsilon_{\alpha,DP}(Q)$ is non-decreasing in $\alpha$ and upper-bounded by $\epsilon_{DP}(Q).$

(9) $\epsilon_{\alpha,M}(Q)$ is non-decreasing in $\alpha$, lower-bounded by $\epsilon_{M}(Q)$ and upper-bounded by $\epsilon_{DP}(Q).$
\end{lem}

The convexity of privacy metrics has long been studied. As we see in Section \ref{sec:3}, the following results spark the ideas in development of optimal mechanisms.

\begin{lem}[Convexity of privacy]\label{lem:con_of_pri}
For fixed $P_X$, in $Q_{Y|X}$,

(1) $\epsilon_{DP}(Q)$, $\epsilon_{\delta}(Q)$ and $\epsilon_{MI}(Q)$ are quasi-convex;

(2) (\cite{Issa2020}) $e^{\epsilon_{ML}(Q)}$ is convex;

(3) (\cite{Ho2015}) $\epsilon_{\alpha,DP}(Q)$ is convex;

(4) (\cite{Ho2015, Verdu2015}) $\frac{1}{\alpha-1}\exp\left({\frac{\alpha-1}{\alpha}\epsilon_{\alpha,M}(Q)}\right)$ is convex and $\epsilon_{\alpha,M}(Q)$ is quasi-convex;

(5) (\cite{Cover2012}) $\epsilon_{M}(Q)$ is convex.
\end{lem}

\begin{proof}[Proof of Lemma \ref{lem:con_of_pri}]
The results of $(2)$-$(5)$ are well-known and one may refer to the publications mentioned above. We now prove $(1)$. Let $Q_1$ and $Q_2$ be $(\epsilon_1,\delta)$ and $(\epsilon_2,\delta)$-differential private mechanisms respectively for some $\delta\in[0,1).$ Then for any $y\in\mathcal{Y}$ and neighbors $x$, $x^{\prime}$ in $\mathcal{X}$, $Q_i(y|x)\leq e^{\epsilon_i}Q_{i}(y|x')+\delta$, $i=1$, $2$. For $0\leq\lambda\leq1$, we have
\beqs
\lefteqn{Q(y|x):=\lambda Q_1(y|x)+(1-\lambda)Q_2(y|x)}\\
&&\leq \lambda \left(e^{\epsilon_1}Q_{1}(y|x')+\delta\right)+(1-\lambda)\left(e^{\epsilon_2}Q_{2}(y|x')+\delta\right)\\
&&\leq e^{\max\{\epsilon_1,\epsilon_2\}}\left(\lambda Q_1(y|x')+(1-\lambda)Q_2(y|x')\right)+\delta\\
&&= e^{\max\{\epsilon_1,\epsilon_2\}}Q(y|x')+\delta.
\eeqs
Thus, $\epsilon_{\delta}(Q)\leq\max\{\epsilon_{\delta}(Q_1),\epsilon_{\delta}(Q_2)\}$.

Next, let $\epsilon_{MI}(Q_i)\leq \epsilon_i$ for $i=1$ and $2$. Then for any $y\in\mathcal{Y}$ and $x\in\mathcal{X}$, we have
\beqs
\lefteqn{Q(y|x):=\lambda Q_1(y|x)+(1-\lambda)Q_2(y|x)}\\
&&\leq \lambda e^{\epsilon_1}\mathbb{E}_{x}[Q_1(y|x)]+(1-\lambda)e^{\epsilon_2}\mathbb{E}_{x}[Q_2(y|x)]\\
&&\leq e^{\max\{\epsilon_1,\epsilon_2\}}\left(\lambda\mathbb{E}_{x}[Q_1(y|x)]+(1-\lambda)\mathbb{E}_{x}[Q_2(y|x)]\right)\\
&&= e^{\max\{\epsilon_1,\epsilon_2\}}\mathbb{E}_{x}[Q(y|x)].
\eeqs
Thus, $\epsilon_{MI}(Q)\leq\max\{\epsilon_{MI}(Q_1),\epsilon_{MI}(Q_2)\}$. This completes the proof.
\end{proof}

The following two lemmas reveal the probability preservation properties of the $\alpha$-type privacy measurements.

\begin{lem}[\cite{Mironov2017}]\label{lem:renyi}
For $\alpha>1$, let $D_{\alpha}(P\Vert Q)$ denote the R\'{e}nyi divergence between $P$ and $Q$ over $\mathcal{R}$. Then we have

(1) Monotonicity. $D_{\alpha}(P\Vert Q)$ is non-decreasing in $\alpha$.

(2) Probability preservation. For any arbitrary event $A\subset\mathcal{R}$, 
$$P(A)\leq \left(e^{D_{\alpha}(P\Vert Q)}Q(A)\right)^{\frac{\alpha-1}{\alpha}}.$$
\end{lem}

\begin{lem}[\cite{Esposito2020a}]\label{lem:sibson}
Let $(\mathcal{X}\times\mathcal{Y},\mathcal{F},P_{XY})$ and $(\mathcal{X}\times\mathcal{Y},\mathcal{F},P_{X}P_{Y})$ be two probability spaces, and assume $P_{XY}\ll P_{X}P_{Y}.$ Given $E\in\mathcal{F}$, let $E_{y}:=\{x:(x,y)\in E\}$. For $\alpha>1$, let $I_{\alpha}(X;Y)$ denote Sibson mutual information. Then
$$P_{XY}(E)\leq\max_{y}\left(P_{X}(E_y)\exp(I_{\alpha}(X;Y))\right)^{\frac{\alpha-1}{\alpha}}.$$
\end{lem}

\subsection{Coloring Scheme}\label{app:cs}
Let $P$ be a prior distribution such that $P_1\geq P_2\geq \cdots \geq P_{m}.$ For any mechanism $Q$, we write it as the matrix form below.
$$
Q = \begin{pmatrix}  
Q(1|1) & Q(2|1) & \cdots & Q(m|1)\\ Q(1|2) & Q(2|2) & \cdots & Q(m|2)\\
\vdots & \vdots & \ddots & \vdots\\
Q(1|m) & Q(2|m) & \cdots & Q(m|m)
\end{pmatrix}.$$
To reinforce the understanding of prior distribution, we define the accumulation function $D^{(0)}=0$ $D^{(k)}:=\sum_{j=m-k+1}^{m}P_{j}$ for $1\leq k\leq m$.

By Lemma \ref{lem:02} and \ref{lem:04}, there exists mechanism costing zero approximate privacy loss as well as mechanism causing no maximal leakage under specific assumption sacrificing accuracy. In this appendix, we assume that privacy loss is inevitable, i.e., the threshold $D$ is at most $(1-\delta)D^{(m-1)}$ for approximate differential privacy and $D^{(m-1)}$ for maximal leakage. 

For any mechanism $Q$ with $\epsilon_{\delta}(Q)>0$, we say two distinct entries $Q(j|i)$ and $Q(j|k)$ of the same column consist of a critical pair if $Q(j|i)=e^{\epsilon(Q)}Q(j|k)+\delta$. We mow color every single entry of $Q$ following the rules below.
\begin{itemize}
    \item Color the element black  if it is the bigger one in the critical pair.
    \item Color the element red if it is the smaller one in the critical pair.
    \item Color the element white, otherwise.
\end{itemize}
Notice that if $Q(j|i)=0$ then for the $j$th column, every element is not greater than $\delta$ with equality if the element is black.

Similarly, for any mechanism $Q$ with $\epsilon_{ML}(Q)>0$, we color every single entry of $Q$ following the rules below.
\begin{itemize}
    \item Color the element black  if it is the biggest among the entries of the same column.
    \item Color the element red if it is the smallest one among the entries of the same column.
    \item Color the element white, otherwise.
\end{itemize}

\subsection{Special Transformers}\label{app:st}
\begin{lem}(Size of the output space)\label{lem:size_of_out}
There exists a mechanism $Q\in\mathcal{Q}^{*}_{\star}(\mathcal{P},D)$ such that $|Q(\mathcal{X})|\leq m$ for the notions aforementioned except for the approximate privacy with positive $\delta$.
\end{lem}

 To prove Lemma \ref{lem:size_of_out}, we introduce the transformation on mechanism given by Kalantari et al\cite{Kalantari2018}.

\begin{proof}[Proof of Lemma \ref{lem:size_of_out}]
For any $Q\in\mathcal{Q}(\mathcal{P},D)$, assume that $|Q(X)|=N> m$, that is $Q(j|i)>0$ for $j\leq N$ and $Q(j|i)=0$ for $j>N$. We create a mechanism $\tilde{Q}$ defined as follows.
$$
\tilde{Q}(j|i)=
\begin{cases}
{Q}(N-1|i)+{Q}(N|i) &\text{if }j=N-1,\\
{Q}(j|i) &\text{if }j<N-1,\\
0 &\text{if }j> N-1.
\end{cases}
$$
Note that $|\tilde{Q}(X)|=N-1$, and for any $P\in\mathcal{P}$, $\mathbb{E}_{P,\tilde{Q}}[d(x,y)]=\sum_{i=1}^{m}P_i(1-\tilde{Q}(i|i))\leq\mathbb{E}_{P,Q}[d(x,y)]$ where the equality holds if and only if $N>m+1$. It is followed by the fact $\tilde{Q}$ belongs to $\mathcal{Q}(\mathcal{P},D)$. Hence, it is sufficient to prove that $\epsilon(\tilde{Q})\leq\epsilon(Q)$.
\begin{enumerate}
\item For $1\leq j\leq N-2$ and $1\leq i\neq k\leq m$, $$\tilde{Q}(j|i)=Q(j|i)\leq e^{\epsilon_1(Q)}Q(j|k)=e^{\epsilon_1(Q)}\tilde{Q}(j|k),$$
\beqs
\lefteqn{\tilde{Q}(N-1|i)=Q(N-1|i)+Q(N|i)}\\
&&\leq e^{\epsilon_{DP}(Q)}Q(N-1|k)+e^{\epsilon_{DP}(Q)}Q(N|k)\\
&&=e^{\epsilon_{DP}(Q)}\tilde{Q}(N-1|k).
\eeqs
Hence, $\epsilon_{DP}(\tilde{Q})\leq\epsilon_{DP}(Q).$
\item For maximal information, one can obtain that $\epsilon_{MI}(\tilde{Q})\leq\epsilon_{MI}(Q)$ by changing $Q(j|k)$ to $\mathbb{E}_{k}[Q(j|k)]$ in $1)$.
\item The inequality 
\beqs
\lefteqn{\mathcal{L}(X\to\tilde{Q}(X))=\sum_{j=1}^{N-1}\max_{1\leq i\leq m}\tilde{Q}(j|i)}\\
&&=\sum_{j=1}^{N-2}\max_{1\leq i\leq m}\tilde{Q}(j|i)+\max_{1\leq i\leq m}\{\tilde{Q}(N-1|i),\tilde{Q}(N|i)\}\\
&&\leq \sum_{j=1}^{N}\max_{1\leq i\leq m}Q(j|i)=\mathcal{L}(X\to{Q}(X))
\eeqs
yields the result $\epsilon_{ML}(\tilde{Q})\leq\epsilon_{ML}(Q).$
\item Let $f_{\alpha}(x,y)=x^{\alpha}y^{1-\alpha}$ for $0\leq x$, $y\leq1$. Then $f_{\alpha}$ is convex since $\nabla^{2}f_{\alpha}$ is positive semi-definite. Therefore, we have for $1\leq i\neq k\leq m$,
\beqs
\lefteqn{\sum_{j=1}^{N-1}\tilde{Q}^{\alpha}(j|i)\tilde{Q}^{1-\alpha}(j|k)}\\
&&=\sum_{j=1}^{N-2}{Q}^{\alpha}(j|i){Q}^{1-\alpha}(j|k)+\\
&&\Big(Q(N-1|i)+Q(N|i)\Big)^{\alpha}\Big(Q(N-1|k)+Q(N|k)\Big)^{1-\alpha}\\
&&\leq \sum_{j=1}^{N-2}{Q}^{\alpha}(j|i){Q}^{1-\alpha}(j|k)+\\
&&{Q}^{\alpha}(N-1|i){Q}^{1-\alpha}(N-1|k)+{Q}^{\alpha}(N|i){Q}^{1-\alpha}(N|k)\\
&&=\sum_{j=1}^{N}{Q}^{\alpha}(j|i){Q}^{1-\alpha}(j|k).
\eeqs
Hence, $\epsilon_{\alpha,DP}(\tilde{Q})\leq\epsilon_{\alpha,DP}(Q).$
\item The function $\frac{1}{\alpha-1}\exp\left({\frac{\alpha-1}{\alpha}\epsilon_{\alpha,M}(Q)}\right)$ is convex (refer to Lemma \ref{lem:con_of_pri} in Appendix \ref{app:bpp}), and thus so is the situation for fixed $y$. Hence, we have
\beqs
\lefteqn{\sum_{j=1}^{N-1}\left(\sum_{i=1}^{m}P_{i}\tilde{Q}^{\alpha}(j|i)\right)^{\frac{1}{\alpha}}}\\
&&=\sum_{j=1}^{N-2}\left(\sum_{i=1}^{m}P_{i}{Q}^{\alpha}(j|i)\right)^{\frac{1}{\alpha}}\\
&&\ \ +\left(\sum_{i=1}^{m}P_{i}\Big(Q(N-1|i)+Q(N|i)\Big)^{\alpha}\right)^{\frac{1}{\alpha}}\\
&&\leq\sum_{j=1}^{N-2}\left(\sum_{i=1}^{m}P_{i}{Q}^{\alpha}(j|i)\right)^{\frac{1}{\alpha}}\\
&&\ \ +\left(\sum_{i=1}^{m}P_{i}{Q}^{\alpha}(N-1|i)\right)^{\frac{1}{\alpha}}+\left(\sum_{i=1}^{m}P_{i}{Q}^{\alpha}(N|i)\right)^{\frac{1}{\alpha}}\\ 
&&=\sum_{j=1}^{N}\left(\sum_{i=1}^{m}P_{i}{Q}^{\alpha}(j|i)\right)^{\frac{1}{\alpha}}
\eeqs
followed by $\epsilon_{\alpha,M}(\tilde{Q})\leq\epsilon_{\alpha,M}(Q).$
\item As $\alpha$ tends to $1$ in $5)$, we have $\epsilon_{M}(\tilde{Q})\leq\epsilon_{M}(Q).$
\end{enumerate}
\end{proof}
This explains why we choose synthetic version as the released database in local privacy framework. Next, inspired by the construction of transformation above, we introduce several interesting transformers which are useful in coloring scheme. For any mechanism $Q$ with $$Q(\sigma(1)|\sigma(1))\geq Q(\sigma(2)|\sigma(2))\geq \cdots \geq Q(\sigma(m)|\sigma(m))$$ for some permutation $\sigma$ over $\{1,2,\cdots,m\}$. We define the transformer $\mathcal{T}_1$ over $Q$ as follows.
$$\mathcal{T}_1:Q\mapsto \mathcal{T}_{1}Q \text{ where }\mathcal{T}_{1}Q(j|i)=Q(\sigma(j)|\sigma(i))$$ for $1\leq i,$ $j\leq m$.
\begin{lem}\label{lem:trans1}
For approximate differential privacy and maximal leakage, if mechanism $Q$ satisfying $\epsilon(Q)>0$, then $\mathcal{T}_{1}Q$ is still a well-defined mechanism such that 
\begin{enumerate}
    \item $\mathcal{T}_{1}Q$ is $(P,\mathbb{E}_{P,Q}[d(x,y)])$-valid for any prior distribution $P$;
    \item $\epsilon(\mathcal{T}_{1}Q)=\epsilon(Q)$;
    \item $\mathcal{T}_{1}Q(1|1)\geq \mathcal{T}_{1}Q(2|2)\geq\cdots\geq\mathcal{T}_{1}Q(m|m).$
\end{enumerate}
\end{lem}

We skip the proof for simplicity. It is worth noting that if we color $\mathcal{T}_{1}Q$ following the rules aforementioned, then  $\mathcal{T}_{1}Q(j|i)$ and $Q(\sigma(j)|\sigma(i))$ have the same color for $1\leq i$, $j\leq m.$

\begin{lem}\label{lem:trans2}
There exists $Q\in\mathcal{Q}^{*}(P,D)$ for approximate differential privacy or maximal leakage, such that all the off-diagonal elements in the same row have the same color.
\end{lem}

\begin{proof}
We denote the number of black (red, white, respectively) off-diagonal elements in the $k$th row by $n^{(k)}_{B}$ ($n^{(k)}_{R}$, $n^{(k)}_{W}$, respectively). 
\begin{enumerate}
    \item If $n^{(k)}_{B}$, $n^{(k)}_{R}>0$, then we define \beqs
    &&Q^{\prime}(j|i)=\\
    &&\begin{cases}
    Q(j|k)-\Delta/n^{(k)}_{B},\text{ if }i=k\neq j\text{ and }Q(j|k)\text{ is black},\\
    Q(j|k)+\Delta/n^{(k)}_{R},\text{ if }i=k\neq j\text{ and }Q(j|k)\text{ is white},\\
    Q(j|i)\text{, otherwise}
    \end{cases}
    \eeqs where $\Delta$ is small enough such that $0\leq Q^{\prime}\leq 1$ and $\epsilon$ does not increase. Moreover, $Q^{\prime}$ is also a $(P,D)$-valid mechanism.
    \item If $n^{(k)}_{B}$, $n^{(k)}_{W}>0$, then we define \beqs
    &&Q^{\prime}(j|i)=\\
    &&\begin{cases}
    Q(j|k)-\Delta/n^{(k)}_{B},\text{ if }i=k\neq j\text{ and }Q(j|k)\text{ is black},\\
    Q(j|k)+\Delta/n^{(k)}_{W},\text{ if }i=k\neq j\text{ and }Q(j|k)\text{ is white},\\
    Q(j|i)\text{, otherwise}
    \end{cases}
    \eeqs where $\Delta$ is small enough such that $0\leq Q^{\prime}\leq 1$ and $\epsilon$ does not increase. Moreover, $Q^{\prime}$ is also a $(P,D)$-valid mechanism.
    \item If $n^{(k)}_{W}$, $n^{(k)}_{R}>0$, then we define \beqs
    &&Q^{\prime}(j|i)=\\
    &&\begin{cases}
    Q(j|k)-\Delta/n^{(k)}_{W},\text{ if }i=k\neq j\text{ and }Q(j|k)\text{ is black},\\
    Q(j|k)+\Delta/n^{(k)}_{R},\text{ if }i=k\neq j\text{ and }Q(j|k)\text{ is white},\\
    Q(j|i)\text{, otherwise}
    \end{cases}
    \eeqs where $\Delta$ is small enough such that $0\leq Q^{\prime}\leq 1$ and $\epsilon$ does not increase. Moreover, $Q^{\prime}$ is also a $(P,D)$-valid mechanism.
\end{enumerate}
Notice the transformation can be done repeatedly until $Q^{\prime}$ satisfies the claim.
\end{proof}

Through the composition of the operators above, we obtain a transformer $\mathcal{T}_2$ such that for any optimal mechanism $Q$,
\begin{enumerate}
    \item $\mathcal{T}_{2}Q$ is $(P,\mathbb{E}_{P,Q}[d(x,y)])$-valid for any prior distribution $P$;
    \item $\epsilon(\mathcal{T}_{2}Q)= \epsilon(Q)$;
    \item $\mathcal{T}_{2}Q(j|j)=Q(j|j)$ for $1\leq j\leq m;$
    \item all the off-diagonal elements of a row in $\mathcal{T}_{2}Q$ have the same color.
\end{enumerate}

\begin{lem}\label{lem:trans3}
There exists $Q\in\mathcal{Q}^{*}(P,D)$ for approximate differential privacy or maximal leakage, such that if a diagonal element is not black, then the off-diagonal elements of the same the row are red.
\end{lem}
\begin{proof}
Let $\tilde{\mathcal{Q}}^{*}(P,D)\subset\mathcal{Q}^{*}(P,D)$ be the set of $(P,D)$-valid mechanisms with smallest distortion. Choose an arbitrary mechanism $Q\in\tilde{\mathcal{Q}}^{*}(P,D)$. Let $K$ denote the set of the line numbers $k$ such that $Q(k|k)$ is not black and $Q(j|k)$ is not red for $j\neq k$.We define $$
    Q^{\prime}(j|i)=
    \begin{cases}
    Q(k|k)+\Delta_{k},\text{ if }j=i=k\text{ and }k\in K,\\
    Q(j|k)-\Delta_{k}/(m-1),\text{ if }j\neq i=k\text{ and }k\in K,\\
    Q(j|i)\text{, otherwise}
    \end{cases}
    $$ where $$\Delta_{k}=\min\left\{\max_{i}Q(k|i)-Q(k|k),n\cdot\min_{j}\left(Q(j|k)-\min_{i}Q(j|i)\right)\right\}.$$ Therefore, $Q^{\prime}$ is a $(P,D)$-valid mechanism with $\epsilon(Q^{\prime})\leq\epsilon(Q).$ Furthermore, $$\mathbb{E}_{P,Q^{\prime}}[d(x,y)]=\mathbb{E}_{P,Q}[d(x,y)]-\sum_{k\in K}P_{k}\Delta_{k}<\mathbb{E}_{P,Q}[d(x,y)]$$ which contradicts the assumption $Q\in\tilde{\mathcal{Q}}^{*}(P,D).$ Thus, for any $Q\in\tilde{\mathcal{Q}}^{*}(P,D),$ if a diagonal element is not black, then the off-diagonal elements of the same the row are red.
\end{proof}

\subsection{Special Mechanism}\label{app:sm}
In this subsection, we consider the properties of $Q\in\mathcal{T}_{1}\mathcal{T}_{2}(\tilde{\mathcal{Q}}^{*}(P,D))$. By Lemma \ref{lem:trans1}, \ref{lem:trans2} and \ref{lem:trans3}, $Q$ satisfies the following properties.
\begin{enumerate}
    \item $Q(1|1)\geq Q(2|2)\geq \cdots \geq Q(m|m).$
    \item All the off-diagonal elements of a row in $Q$ have the same color.
    \item If a diagonal element is not black, then the off-diagonal elements of the same the row are red.
\end{enumerate}
There are more traits for the color pattern of $Q$.

\begin{prop}\label{prop:1}
The elements of a row in $Q$ are not all black nor red.
\end{prop}
\begin{proof}
Assume that all elements of the $k$th row are black. Then, there exists a row, denoted by $i$th row, such that $Q(j|k)>Q(j|i)$ for $j\neq i$. Hence, we have $$1=Q(i|k)+\sum_{j\neq i}^{m}Q(j|k)>Q(i|i)+\sum_{j\neq i}^{m}Q(j|i)=1.$$ Contradiction. The proof for no row with all red elements is similar.
\end{proof}

By Property \ref{prop:1}, there exists no black off-diagonal element. Moreover, if a diagonal element is not black, then it is white.
\begin{prop}\label{prop:2}
There is only one possible non-black diagonal element in $Q$.
\end{prop}
\begin{proof}
Assume for $1\leq i<k\leq m$, $Q(i|i)$ and $Q(k|k)$ are not black. Then by Lemma \ref{lem:trans3}, $Q(i|k)$ is red. Thus, $Q(i|l)$ is black for some $1\leq l\leq m$. Because $Q(i|i)$ is not black, $l\neq i$. Thus $Q(i|l)$ is a black off-diagonal element which contradicts Property \ref{prop:1}.
\end{proof}

\begin{prop}\label{prop:3}
For the coloring scheme of approximate differential privacy, if $Q(k|k)$ is the only possible non-black diagonal element, then $Q(k|k)=Q(1|1)$.
\end{prop}
\begin{proof}
By Lemma \ref{lem:trans3}, $Q(j|k)$ is red for $j\neq k$. By Property \ref{prop:2}, $Q(i|i)$ is black for $i\neq k$. Then by Property \ref{prop:1}, $Q(j|i)$ is not black for $j\neq i$. Therefore, $Q(j|i)\geq Q(j|k)$ for $j\neq$ $i$ and $k$.
\beqs
&&\sum_{j=1}^{m}Q(j|i)=\sum_{j=1}^{m}Q(j|k)=1\\
&\Rightarrow& Q(k|k)+Q(k|i)\geq Q(i|i)+Q(k|i)\\
&\Rightarrow& Q(k|k)\geq \delta+Q(k|i)\geq\delta.
\eeqs
Because $Q(k|k)$ is white, there exists no line number $j$ such that $Q(k|j)=0$. Thus, $Q(k|i)>0.$ If $Q(i|k)=0$ then $Q(k|k)\geq Q(i|i).$ If $Q(i|k)>0$ then $\frac{Q(i|i)-\delta}{Q(i|k)}>\frac{Q(k|k)-\delta}{Q(k|i)}.$ Assume that $Q(k|k)<Q(i|i).$ Then $0<Q(i|i)-Q(k|k)\leq Q(i|k)-Q(k|i).$ Thus $$\frac{Q(i|i)-Q(k|k)}{Q(i|k)-Q(k|i)}\leq1.$$ On the contrary, 
\beqs
\lefteqn{\frac{Q(i|i)-Q(k|k)}{Q(i|k)-Q(k|i)}=\frac{Q(i|i)-\delta-(Q(k|k)-\delta)}{Q(i|k)-Q(k|i)}}\\
&&>\frac{Q(i|i)-\delta-\frac{Q(k|i)}{Q(i|k)}(Q(i|i)-\delta)}{Q(i|k)-Q(k|i)}\\
&&=\frac{Q(i|i)-\delta}{Q(i|k)}>1.
\eeqs
Contradiction! Hence, $Q(k|k)\geq Q(i|i)$ for $1\leq i\leq m.$
\end{proof}

\begin{prop}\label{prop:4}
For the coloring scheme of approximate differential privacy, all the off-diagonal elements in the first row of $Q$ are red.
\end{prop}
\begin{proof}
Assume that all the off-diagonal elements in the $k$th row are red and all the off-diagonal elements in the $i$th row are white. Since the sum of each row is one, we have
$$Q(k|k)+Q(k|i)\geq Q(i|i)+Q(k|i).$$
Similar with the proof of Property \ref{prop:3}, this inequality implies $Q(k|k)\geq Q(i|i)$. In other words, all the red rows are arranged above the white ones in $Q$.
\end{proof}

\begin{prop}\label{prop:5}
For the coloring scheme of maximal leakage, all the diagonal elements are black in $Q$.
\end{prop}
\begin{proof}
We claim that if the only possible white diagonal element is $Q(k|k)$, then $Q(j|k)=0$ for $j\neq k.$ Otherwise, we can transfer some non-zero off-diagonal elements' value of the $k$th row to $Q(k|k)$ until it turns black, resulting in a optima l mechanism with smaller distortion. Contradiction to $Q\in\tilde{\mathcal{Q}}^{*}(P,D).$ Thus $Q(j|k)=0$ for $j\neq k$, leading to the fact $Q(k|k)=1$. This is not possible since $Q(k|k)$ is white.
\end{proof}

Combing all aforementioned properties of $Q$, one can obtain Lemma \ref{lem:eq2}.

\subsection{Proof of Theorem \ref{thm:DP4}}\label{app:ml}
The condition of equivalence for zero maximal leakage is given as follows
\begin{lem}\label{lem:04}
For given $P$ such that $P_1\geq P_2\geq \cdots \geq P_m$, $\epsilon^{*}_{ML}(P,D)$ $=$ $0$ if and only if $D\geq D^{(m-1)}.$
\end{lem}

\begin{proof}[Proof of Lemma \ref{lem:04}]
Necessity. If $\epsilon^{*}_{ML}(P,D)=0$ then there exists a mechanism $Q$ such that $\epsilon_{ML}(Q)=0$ and $\mathbb{E}_{P,Q}[d(x,y)]\leq D$. Thus, $\sum_{j}\max_{i}Q(j|i)=1$. Hence, $\sum_{j}Q(j|j)\leq1.$ The expected distortion 
\beqs
\lefteqn{\mathbb{E}_{P,Q}[d(x,y)]=1-\sum_{j=1}^{m}P_{j}Q(j|j)}\\
&&\geq 1-P_{1}\sum_{j=1}^{m}Q(j|j)\geq D^{(m-1)}.
\eeqs
Thus $D^{(m-1)}\leq\mathbb{E}_{P,Q}[d(x,y)]\leq D.$

Sufficiency. For $D\geq D^{(m-1)}$, we define mechanism $Q$ as follows.
$$
Q=
\begin{pmatrix}  
1 & 0 & \cdots & 0\\  
1 & 0 & \cdots & 0\\
\vdots & \vdots & \ddots & \vdots\\
1 & 0 & \cdots & 0
\end{pmatrix}.
$$
One can easily obtain that $\epsilon_{ML}(Q)=0$ and $\mathbb{E}_{P,Q}[d(x,y)]=D^{(m-1)}\leq D.$ Thus, $\epsilon^{*}_{\delta}(P,D)=0$.
\end{proof}

Combing all properties of $Q$ in Appendix \ref{app:sm}, one can obtain the following lemma.
\begin{lem}\label{lem:eq4}
If $\epsilon^{*}_{ML}(P,D)>0$ then there exists a mechanism $Q\in\mathcal{Q}_{ML}^{*}(P,D)$ such that 

(1) $0\leq Q(m|m)\leq Q(m-1|m-1)\leq\cdots\leq Q(1|1)\leq1;$

(2) $\max_{i}Q(j|i)=Q(j|j)$ for $1\leq j\leq m$.
\end{lem}
Combining theses two lemmas leads to the proof of Theorem \ref{thm:DP4}.
\begin{proof}[Proof of Theorem\ref{thm:DP4}]
For $D\geq D^{(m-1)}$, Lemma \ref{lem:04} proves the statement.

For $0<D<D^{(m-1)}$, let $Q^{*}$ be a mechanism satisfying the conditions in Lemma \ref{lem:eq4}. Then $\epsilon_{ML}^{*}(P,D)=\log\sum_{j=1}^{m}Q^{*}(j|j).$ Consider the following optimization problem,
\begin{eqnarray}
\nonumber\lefteqn{d_6:=\min_{\alpha_1,\alpha_2,\cdots,\alpha_m}\log\sum_{j=1}^{m}\alpha_{j}}\\
\nonumber&&\text{subject to}\\
\label{op:ml1}&&(1)\text{ }\sum_{j=1}^{m}P_j\alpha_j\geq1-D,\\
\nonumber&&(2)\text{ }\sum_{j=1}^{m}\alpha_j\geq1,\\
\nonumber&&(3)\text{ }\alpha_{j}\leq\alpha_{j-1}\text{ for }1\leq j\leq m+1
\end{eqnarray}
where $\alpha_0=1$ and $\alpha_{m+1}=0$. Notice that $\{Q^{*}(j|j)\}_{1\leq j\leq m}$ is a feasible solution of \eqref{op:ml1}, and hence $\epsilon_{ML}^{*}(P,D)\geq d_6.$ Conversely, let $\{\alpha_{1}^{*},\alpha_{2}^{*},\cdots,\alpha_{m}^{*}\}$ be the optimal solution of \eqref{op:ml1}. We define the following mechanism.
$$
Q(j|i)=
\begin{cases}
\alpha_{j}^{*}&\text{ if }i=j,\\
\frac{(1-\alpha_{j}^{*})\alpha^{*}_{i}}{\sum_{i\neq j}\alpha_{i}}&\text{ if }i\neq j.
\end{cases}
$$ 
For any $j\in[1,m]$, $\sum_{i=1}^{m}Q(j|i)=1$ which leads to the fact that $Q$ is a conditional probability. The first restricted condition in the scenario of \eqref{op:ml1} implies that $Q$ is within $\mathcal{Q}(P,D)$. The other two imply that $$\epsilon_{ML}^{*}(P,D)\leq\epsilon_{ML}(Q)=\log\sum_{j=1}^{m}\max_{1\leq i\leq m}Q(j|i)=\log\sum_{j=1}^{m}\alpha^{*}_{j}=d_{6}.$$ 

By letting $\beta_{j}:=\alpha_{j}-\alpha_{j+1}$ for $1\leq j\leq m$, we obtain the following standard linear programming.
\begin{eqnarray}
\nonumber\lefteqn{d_7(\alpha_1):=\max_{\beta_1,\beta_2,\cdots,\beta_{m}}-\sum_{j=1}^{m}j\beta_j}\\
\nonumber&&\text{subject to}\\
\nonumber&&(1)\text{ }\sum_{j=1}^{m}(D^{(m-j)}-1)\beta_j\leq-1+D,\\
\label{op:ml2}&&(2)\text{ }-\sum_{j=1}^{m}j\beta_j\leq -1,\\
\nonumber&&(3)\text{ }\sum_{j=1}^{m}\beta_{j}\leq 1,\\
\nonumber&&(4)\text{ }\beta_{j}\geq0\text{ for }1\leq j\leq m.
\end{eqnarray}
Moreover, $d_6=\log(-d_7).$ The dual of \eqref{op:ml2} is the minimization linear program as below.
\begin{eqnarray}
\nonumber\lefteqn{\tilde{d}_7:=\min_{\gamma_1,\gamma_2,\gamma_{3}}(-1+D)\gamma_1-\gamma_2+\gamma_3}\\
\label{op:ml3}&&\text{subject to}\\
\nonumber&&(1)\text{ }D^{(m-j)}\gamma_1-j\gamma_2+\gamma_3\geq -j\text{ for }1\leq j\leq m\\
\nonumber&&(2)\text{ }\gamma_1,\gamma_2,\gamma_{3}\geq0.
\end{eqnarray}
By strong duality theorem, we have ${d}_7=\tilde{d}_7$. Similar with the proof of Theorem \ref{thm:DP2}, we can obtain the following result by fixing $\gamma_2\geq0$.
$$\tilde{d}_7=\min\left\{-1,\frac{D-D^{(k-1)}}{P_{m+1-k}}-(m+1-k)\right\}$$
if $D^{(k-1)}<D\leq D^{(k)}$ and $1\leq k\leq m$. Thus, 
$$\epsilon^{*}_{ML}=\max\left\{0,\log\left(m-k-\frac{D-D^{(k)}}{D^{(k)}-D^{(k-1)}}\right)\right\}$$
if $D^{(k-1)}<D\leq D^{(k)}$ and $1\leq k\leq m$.
\end{proof}

\end{appendix}

\end{document}